\documentclass{llncs}

\usepackage{soul}
\usepackage{subcaption}
\usepackage{paralist}
\usepackage{fullpage}

\usepackage{hyperref}
\usepackage{amsmath,amsfonts,amssymb}
\usepackage{mathtools}
\usepackage[capitalise,nameinlink]{cleveref}

\renewcommand{\orcidID}[1]{\href{https://orcid.org/#1}{\includegraphics[scale=.03]{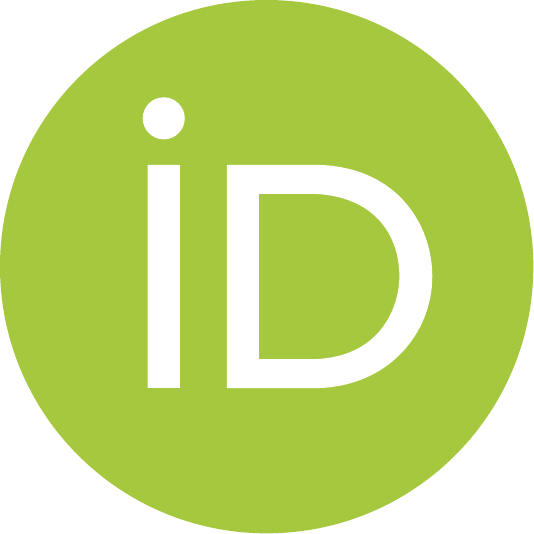}}}

\newtheorem{observation}[theorem]{Observation}
\newtheorem{claimx}[theorem]{Claim}
\newtheorem{lemmax}[theorem]{Lemma}

\def\qed{\ifmmode\squareforqed\else{\unskip\nobreak\hfil
\penalty50\hskip1em\null\nobreak\hfil\squareforqed
\parfillskip=0pt\finalhyphendemerits=0\endgraf}\fi}

\usepackage{thmtools,thm-restate}
\usepackage{microtype}
\usepackage{fullpage}

\author{
    Carlos Alegr\'{i}a\inst{}\orcidID{0000-0001-5512-5298} \and
    Susanna Caroppo\orcidID{0009-0001-4538-8198} \and 
    Giordano {Da Lozzo}\orcidID{0000-0003-2396-5174} \and
    Marco {D'Elia}\orcidID{0009-0008-6266-3324} \and \\
    Giuseppe {Di Battista}\orcidID{0000-0003-4224-1550} \and
    Fabrizio Frati\orcidID{0000-0001-5987-8713} \and
    Fabrizio Grosso\orcidID{0000-0002-5766-4567} \and
    Maurizio Patrignani\orcidID{0000-0001-9806-7411}}

\institute
{
  Department of Engineering, Roma Tre University, Rome, Italy\\
  \email{carlos.alegria@uniroma3.it},
  \email{susanna.caroppo@uniroma3.it},
  \email{marco.delia@uniroma3.it},
  \email{giordano.dalozzo@uniroma3.it},
  \email{giuseppe.dibattista@uniroma3.it},
  \email{fabrizio.frati@uniroma3.it},
  \email{maurizio.patrignani@uniroma3.it},\\
  CeDiPa, Perugia University, Perugia, Italy\\ \email{fabrizio.grosso@unipg.it},
}

\authorrunning{Alegr{\'{\i}}a et al.}

\usepackage{hyperref}
\usepackage[linewidth=1pt]{mdframed}

\Crefname{property}{Property}{Properties}
\Crefname{observation}{Observation}{Observations}
\Crefname{theorem}{Theorem}{Theorems}
\Crefname{section}{Section}{Sections}
\Crefname{figure}{Figure}{Figures}
\Crefname{claimx}{Claim}{Claims}
\Crefname{figure}{Fig.}{Figs.}
\Crefname{section}{Sec.}{Sec.}

\let\doendproof\endproof
\renewcommand\endproof{~\hfill\qed\doendproof}

\newcommand{\UPSE}{{\sc UPSE}\xspace}
\newcommand{\UPSEs}{{\sc UPSE}s\xspace}

\usepackage{graphicx} %
\usepackage{xcolor,soul}

\usepackage{xspace}
\usepackage{todonotes}
\usepackage{complexity}
\usepackage{framed}
\usepackage{tabularx}

\usepackage{fancyhdr}
\usepackage{tikz}
\usetikzlibrary{calc}

\newlength{\RoundedBoxWidth}
\newsavebox{\GrayRoundedBox}
\newenvironment{GrayBox}[1]%
   {\setlength{\RoundedBoxWidth}{.93\columnwidth}
    \def\boxheading{#1}
    \begin{lrbox}{\GrayRoundedBox}
       \begin{minipage}{\RoundedBoxWidth}}%
   {   \end{minipage}
    \end{lrbox}
    \begin{center}
    \begin{tikzpicture}%
       \node(Text)[draw=black!20,fill=white,rounded corners,inner xsep=2ex,inner ysep=2ex,text width=\RoundedBoxWidth]
             {\usebox{\GrayRoundedBox}};
        \coordinate(x) at (current bounding box.north west);
        \node [draw=white,rectangle,inner sep=3pt,anchor=north west,fill=white]
        at ($(x)+(6pt,.75em)$) {\boxheading};
    \end{tikzpicture}
    \end{center}}

\newenvironment{defproblemx}[2]{\noindent\ignorespaces%
                                \FrameSep=6pt%
                                \parindent=6pt
                \vspace{-3mm}            
                \begin{GrayBox}{#1}%
                \begin{tabular*}{\columnwidth}{!{\extracolsep{\fill}}@{\hspace{.1em}} >{\itshape} p{#2} p{0.84\columnwidth} @{}}%
            }{\\[-1.5ex]
                \end{tabular*}%
                \end{GrayBox}%
                \ignorespacesafterend
                \vspace{-4mm}
            }

\newcommand{\problemQuestion}[3]{%
  \begin{defproblemx}{#1}{1.5cm}
    Input: & #2 \\
    Question: & #3
  \end{defproblemx}
}

\DeclareMathOperator{\conv}{\mathcal{CH}}

\DeclareMathOperator{\convR}{\mathcal{H}_R}
\DeclareMathOperator{\convL}{\mathcal{H}_L}

\definecolor{defblue}{rgb}{0.121,0.47,0.705}
\definecolor{darkred}{rgb}{.5,0.27,0.27}
\DeclareTextFontCommand{\emph}{\color{defblue}\em}

\definecolor{lipicsblue}{rgb}{0.08235294118,0.3098039216,0.537254902}

\definecolor{linkblue}{rgb}{0.098,0.098,0.4392}
\definecolor{ourgreen}{rgb}{0.509,0.745,0.235}
\definecolor{ourgreen2}{rgb}{0.409,0.600,0.200}

\definecolor{indianred}{rgb}{0.804,0.361,0.361}
\definecolor{indianred1}{rgb}{1,0.416,0.416}
\definecolor{indianred3}{rgb}{0.804,0.333,0.333}
\definecolor{orangered}{rgb}{1,0.271,0}
\definecolor{coral1}{rgb}{1,0.447,0.337}
\definecolor{rosybrown2}{rgb}{0.933,0.231,0.231}
\definecolor{aquamarine4}{rgb}{0.271,0.545,0.455}
\definecolor{chartreuse3}{rgb}{0.4,0.804,0}
\definecolor{mediumpurple3}{rgb}{0.537,0.408,0.804}
\definecolor{mediumvioletred}{rgb}{0.78,0.082, 0.522}

\hypersetup{colorlinks=true,
    linkcolor=orangered,
    anchorcolor=lipicsblue,
    citecolor=lipicsblue,
    filecolor=lipicsblue,
    menucolor=lipicsblue,
    urlcolor=lipicsblue,
    bookmarksopen=true,
    bookmarksopenlevel=2,
    bookmarksnumbered=true,
    plainpages=false,
    }

\newcommand{\bigO}{\mathcal{O}}

\newcommand{\Tparam}[1]{\ensuremath{T[\hspace{0.3mm}{#1}\hspace{0.3mm}]}}
\newcommand{\Qparam}[1]{\ensuremath{Q[\hspace{0.3mm}{#1}\hspace{0.3mm}]}}

\date{}

\title{Upward Pointset Embeddings of Planar $st$-Graphs\thanks{This research was supported, in part,  by MUR of Italy (PRIN Project no.~2022ME9Z78~-- NextGRAAL and PRIN Project no.~2022TS4Y3N~-- EXPAND). A preliminary version of this work appeared in~\cite{DBLP:conf/gd/AlegriaCLDBFGP24}.}
}

\begin{document}

\maketitle

\begin{abstract}
We study upward pointset embeddings (\UPSEs) of planar $st$-graphs. 
Let $G$ be a planar $st$-graph and let  $S \subset \mathbb{R}^2$ be a pointset with $|S|= |V(G)|$.
An {\em UPSE} of $G$ on $S$ is an upward planar straight-line drawing of $G$ that maps the vertices of $G$ to the points of $S$.
We consider both the problem of testing the existence of an \UPSE of $G$ on $S$ ({\sc UPSE Testing}) and the problem of enumerating all \UPSEs of $G$ on $S$.
We prove that {\sc UPSE Testing} is \NP-complete even for $st$-graphs that consist of a set of directed $st$-paths sharing only $s$ and $t$.
On the other hand, if $G$ is an $n$-vertex planar $st$-graph whose maximum $st$-cutset has size $k$, then {\sc UPSE Testing} can be solved in $\bigO(n^{4k})$ time with $\bigO(n^{3k})$ space; also, all the \UPSEs of $G$ on $S$ can be enumerated with $\bigO(n)$ worst-case delay, using $\bigO(k  n^{4k} \log n)$ space, after $\bigO(k  n^{4k} \log n)$ set-up time.
Moreover, for an $n$-vertex $st$-graph whose underlying graph is a cycle, we provide a necessary and sufficient condition for the existence of an \UPSE on a given pointset, which can be tested in $\bigO(n \log n)$ time. Related to this result, we give an algorithm that, for a set $S$ of $n$ points,  enumerates all the non-crossing monotone Hamiltonian cycles on $S$ with $\bigO(n)$ worst-case delay, using $\bigO(n^2)$ space, after $\bigO(n^2)$ set-up time.
\end{abstract}

\section{Introduction}
Given an $n$-vertex upward planar graph $G$ and a set $S$ of $n$ points in the plane, an \emph{upward pointset embedding} ({\sc UPSE}) of $G$ on $S$ is an upward planar drawing of $G$ where the vertices are mapped to the points of $S$ and the edges are represented as straight-line segments.
The {\sc Upward Pointset Embeddability Testing Problem} ({\sc UPSE Testing}) asks whether an upward planar graph $G$ has an UPSE on a given pointset $S$. 

Pointset embedding problems are classic challenges in Graph Drawing and have been considered for both undirected and directed graphs. For an undirected graph, a \emph{pointset embedding} ({\sc PSE}) has the same definition of an UPSE, except that the drawing must be planar, rather than upward planar.
The {\sc Pointset Embeddability Testing Problem} ({\sc PSE Testing}) asks whether a planar graph has a PSE on a given pointset~$S$. Pointset embeddings have been studied by several authors. It is known that a graph admits a PSE on {\em every} pointset in general position if and only if it is outerplanar~\cite{castaneda1996straight,gritzmann1991e3341}; such a PSE can be constructed efficiently~\cite{DBLP:conf/gd/Bose97,bose2002-outerplanar,DBLP:conf/gd/BoseMS95,bose1997-trees}.  {\sc PSE Testing} is, in general, \NP-complete~\cite{cabello2006planar}, however it is polynomial-time solvable if the input graph is a planar $3$-tree~\cite{DBLP:conf/gd/NishatMR10,nishat2012point}. More in general, a polynomial-time algorithm for {\sc PSE Testing} exists if the input graph has a fixed embedding, bounded treewidth, and bounded face size~\cite{biedl2012point}. PSE becomes \NP-complete if one of the latter two conditions does not hold. PSEs have been studied also for dynamic graphs~\cite{dibattista2022small}.

The literature on UPSEs is not any less rich than the one on PSEs. From a combinatorial perspective, the directed graphs with an UPSE on a one-sided convex pointset have been characterized~\cite{binucci2010upward,DBLP:journals/siamcomp/HeathPT99}; all directed trees are among them. Conversely, there exist directed trees that admit no UPSE on certain convex pointsets~\cite{binucci2010upward}. Directed graphs that admit an UPSE on any convex pointset, but not on any pointset in general position, exist~\cite{angelini2010upward}. It is still unknown whether every digraph whose underlying graph is a path admits an UPSE on every pointset in general position, see, e.g.,~\cite{DBLP:journals/comgeo/Mchedlidze13}. UPSEs where bends along the edges are allowed have been studied in~\cite{binucci2010upward,digiacomo20241,giordano2015bends,DBLP:conf/gd/KaufmannW99,kaufmann2002bends}. From the computational complexity point of view~\cite{DBLP:conf/gd/KaufmannMS11,kaufmann2013upward}, it is known that {\sc UPSE Testing} is \NP-hard, even for planar $st$-graphs and $2$-convex pointsets, and that {\sc UPSE Testing} can be solved in polynomial time if the given pointset is convex.

{\bf Our contributions.} We tackle {\sc UPSE Testing} for planar $st$-graphs. Planar $st$-graphs constitute an important class of upward planar graphs; indeed, it is known that every upward planar graph is a subgraph of a planar $st$-graph~\cite{DBLP:journals/tcs/BattistaT88}. Let $G$ be an $n$-vertex planar $st$-graph and $S$ be a set of $n$ points in the plane.  We adopt the common assumption in the context of upward pointset embeddability, see e.g.~\cite{angelini2010upward,binucci2010upward,DBLP:conf/gd/KaufmannMS11,kaufmann2013upward}, that no two points of $S$ lie on the same horizontal line. 
Our results are the following:
\begin{itemize}
\item In \cref{se:hardness}, we show that {\sc UPSE Testing} is \NP-hard even if $G$ consists of a set of internally-disjoint $st$-paths (\cref{th:st-hardness}). A similar proof shows that {\sc UPSE Testing} is \NP-hard for directed trees consisting of a set of directed root-to-leaf paths (\cref{th:tree-hardness}). This answers an open question from~\cite{arseneva2021upward} and strengthens a result therein, which shows \NP-hardness for directed trees with multiple sources and with a prescribed mapping for a vertex. 
\item In \cref{sec:k-paths}, we show that {\sc UPSE Testing} can be solved in $\bigO{(n^{4k})}$ time and $\bigO(n^{3k})$ space, where $k$ is the size of the largest $st$-cutset of $G$ (\cref{th:k-paths}). This parameter measures the ``fatness'' of the digraph and coincides with the length of the longest directed path in the dual~\cite{DBLP:journals/tcs/BattistaT88}. By leveraging on the techniques developed for the UPSE testing algorithm, we also show how to enumerate all \UPSEs of $G$ on $S$ with $\bigO(n)$ worst-case delay, using $\bigO(k  n^{4k} \log n)$ space, after $\bigO(k  n^{4k} \log n)$ set-up time (\cref{th:st-enumeration}). Similarly to previous algorithms for pointset embeddings~\cite{biedl2012point,kaufmann2013upward}, our algorithms are based on dynamic programming; however, our algorithms employ an explicit correspondence between a structure in the graph (an $st$-cutset) and a structure in the pointset (a cut defined by a horizontal line), which might be of interest.
\item In \cref{sec:2-paths}, we provide a simple characterization of the pointsets in general position that allow for an UPSE of $G$, if $G$ consists of two (internally-disjoint) $st$-paths. Based on that, we provide an $\bigO(n\log n)$ testing algorithm for this case (\cref{th:two-paths}). Previously, a characterization of the directed graphs admitting an UPSE on a given pointset was known only if the pointset is one-sided convex~\cite{binucci2010upward,DBLP:journals/siamcomp/HeathPT99}.
\item Finally, in \cref{se:geometry}, inspired by the fact that an UPSE of a planar $st$-graph composed of two $st$-paths defines a non-crossing monotone Hamiltonian cycle on $S$, we provide an algorithm that enumerates all the non-crossing monotone Hamiltonian cycles on a given pointset with $\bigO(n)$ worst-case delay, and $\bigO(n^2)$ space usage and set-up time (\cref{th:hamiltonian-enumeration}). 
\end{itemize}

Concerning our last result, we remark that a large body of research has considered problems related to enumerating and counting non-crossing structures on a given pointset~\cite{DBLP:journals/dcg/AlvarezBCR15,DBLP:journals/dam/CheonCES22,DBLP:journals/dm/FlajoletN99,DBLP:conf/compgeom/MarxM16,DBLP:conf/birthday/RazenW11}. Despite this effort, the complexity of counting the non-crossing Hamiltonian cycles, often called \emph{polygonalizations}, remains open~\cite{DBLP:journals/dcg/Eppstein20a,DBLP:conf/compgeom/MarxM16,DBLP:journals/ijcga/MitchellO01}. However, it is possible to enumerate all polygonalizations of a given pointset in singly-exponential time~\cite{DBLP:journals/jocg/Wettstein17,DBLP:journals/dam/YamanakaAHOUY21}. Recently, an algorithm has been shown~\cite{eppstein2023non} to enumerate all polygonalizations of a given pointset in time polynomial in the output size, i.e., bounded by a polynomial in the number of solutions. However, an enumeration algorithm with  polynomial (in the input size) delay is not yet known, neither in the worst-case nor in the average-case acception. Our enumeration algorithm achieves this goal for the case of monotone polygonalizations. 

We also remark that the enumeration of graph drawings has been recently considered in~\cite{efficient-enumeration}.

\section{Preliminaries}\label{sec:preliminaries}

We use standard terminology in graph theory~\cite{Diestelbook} and graph drawing~\cite{BattistaETT99}. For an integer $k>0$, let $[k]$ denote the set $\{1,\dots,k\}$. A \emph{permutation with repetitions} of $k$ elements from $U$ is an arrangement of any $k$ elements of a set $U$, where repetitions are allowed.

For a point $p \in \mathbb{R}^2$, we denote by $x(p)$ and $y(p)$ the $x$- and $y$-coordinate of $p$, respectively.
The \emph{convex hull} $\conv(S)$ of a set $S$ of points in ${\mathbb R}^2$ is the union of all convex combinations of points in $S$. The \emph{boundary} $ \mathcal{B}(S)$ of $\conv(S)$ is the polygon with minimum perimeter enclosing~$S$.
The points of $S$ with lowest and highest $y$-coordinates are the \emph{south} and \emph{north extreme} of~$S$, respectively; we also refer to them as to the \emph{extremes} of~$S$.
The \emph{left envelope} of $S$ is the subpath~$\mathcal{E}_L(S)$ of $\mathcal{B}(S)$ that lies to the left of the line passing through the extremes of $S$; it includes the extremes of $S$. The \emph{right envelope} $\mathcal{E}_R(S)$ of $S$ is defined analogously. We denote the subset of $S$ in $\mathcal{E}_L(S)$ and in $\mathcal{E}_R(S)$ by $\convL(S)$ and $\convR(S)$, respectively.
A \emph{polyline} $(p_1,\dots, p_k)$, with $k\geq 2$, is a chain of straight-line segments.

We call \emph{ray} any of the two half-lines obtained by cutting a straight line at any of its points, which is the \emph{starting point} of the ray. A ray is \emph{upward} if it passes through points whose $y$-coordinate is larger than the one of the starting point of the ray. We denote by $\rho(p,q)$ the ray starting at a point $p$ and passing through a point $q$. For a set of points $S$ and a point~$p$ whose $y$-coordinate is smaller than the one of every point in $S$, we denote by $\rho(p,S)$ the rightmost upward ray starting at $p$ and passing through a point of $S$. That is,  the clockwise  rotation around $p$ which brings $\rho(p,S)$ to coincide with any other upward ray starting at $p$ and passing through a point of $S$ is larger than $180^\circ$. Analogously, we denote by $\ell(p,S)$ the leftmost upward ray starting at $p$ and passing through a point of $S$.

A polyline $(p_1,\dots,p_k)$ is \emph{$y$-monotone} if $y(p_i)<y(p_{i+1})$, for $i=1,\dots,k-1$. A \emph{monotone path} on a pointset $S$ is a $y$-monotone polyline $(p_1,\dots,p_k)$ such that the points $p_1,\dots,p_k$ belong to $S$. A \emph{monotone cycle} on $S$ consists of two monotone paths on $S$ that share their endpoints. A \emph{monotone Hamiltonian cycle} $(p_1,\dots,p_k,p_1)$ on $S$ is a monotone cycle on $S$ such that each point of $S$ is a point $p_i$ (and vice versa). 

A path $(v_1,\dots,v_k)$ is \emph{directed} if, for $i=1,\dots,k-1$, the edge $(v_i,v_{i+1})$ is directed from $v_i$ to $v_{i+1}$; the vertices $v_2,\dots,v_{k-1}$ are \emph{internal}. A \emph{planar $st$-graph} is an acyclic digraph with one source $s$ and one sink $t$, which admits a planar embedding in which $s$ and $t$ are on the boundary of the outer face. An \emph{$st$-path} in a planar $st$-graph is a directed path from $s$ to $t$. A drawing of a directed graph is \emph{straight-line} if each edge is represented by a straight-line segment, it is \emph{planar} if no two edges cross, and it is \emph{upward} if every edge is represented by a Jordan arc 
monotonically increasing along the $y$-axis from the tail to the head. A digraph that admits an upward planar drawing is an \emph{upward planar graph}. Every upward planar graph admits an upward planar straight-line drawing~\cite{DBLP:journals/tcs/BattistaT88}. An \emph{Upward Pointset Embedding} (\emph{UPSE}, for short) of an upward planar graph~$G$ on a pointset $S$ is an upward planar straight-line drawing of~$G$ that maps each vertex of~$G$ to a point in $S$. In this paper, we study the following problem.

\problemQuestion{\sc Upward Pointset Embeddability Testing Problem (UPSE Testing)}%
{An $n$-vertex upward planar graph $G$ and a pointset $S \subset \mathbb{R}^2$ with $|S| = n$.}%
{Does there exist an \UPSE of $G$ on $S$?}  
\medskip

In the remainder, we assume that not all points in $S$ lie on the same line, as otherwise there is an UPSE if and only if the input is a directed path. Recall that no two points in $S$ have the same $y$-coordinate. Unless otherwise specified, we do not require points to be in \emph{general position}, i.e., we allow three or more points to lie on the same line.

\section{NP-Completeness of UPSE Testing}\label{se:hardness}

In this section we prove that {\sc UPSE Testing} is \NP-complete. The membership in \NP~is obvious, as one can non-deterministically assign the vertices of the input graph $G$ to the points of the input pointset $S$ and then test in polynomial time whether the assignment results in an upward planar straight-line drawing of $G$. In the remainder of the section, we prove that {\sc UPSE Testing} is \NP-hard even in very restricted cases.

We first show a reduction from \textsc{3-Partition} to instances of UPSE in which the input is a planar $st$-graph composed of a set of internally-disjoint $st$-paths. An instance of \textsc{3-Partition} consists of a set $A =\{a_1, \dots, a_{3b}\}$ of $3b$ integers, where $\sum_{i=1}^{3b} a_i = bB$ and $B/4 \leq a_i \leq B/2$, for $i = 1, \dots, 3b$.
The \textsc{3-Partition} problem asks whether $A$ can be partitioned into $b$ subsets $A_1, \dots, A_b$, each with three integers, so that the sum of the integers in each set $A_i$ is $B$. 
For example, an instance of \textsc{3-Partition} might be a set $A=\{2,2,2,2,2,2,3,3,3,3,4,4\}$, with $B=8$ and $b=4$. The instance is positive, as certified by the sets $A_1=\{2,2,4\}$, $A_2=\{2,2,4\}$, $A_3=\{2,3,3\}$, and $A_4=\{2,3,3\}$. Since \textsc{3-Partition} is strongly \NP-hard~\cite{garey1979computers}, we may assume that $B$ is bounded by a polynomial function of $b$.
Given an instance $A$ of \textsc{3-Partition}, we show how to construct in polynomial time, precisely $\bigO(b\cdot B)$, an equivalent instance $(G, S)$ of \mbox{\sc UPSE Testing}. 

The $n$-vertex planar $st$-graph $G$ is composed of $4b+1$ internally-disjoint $st$-paths. Namely, for $i = 1, \dots, 3b$, we have that $G$ contains an \emph{$a_i$-path}, i.e., a path with $a_i$ internal vertices, and $b+1$ additional $k$-paths, where $k=2B+1$. Note that $n=2+(b+1)k+\sum_{i=1}^{3b} a_i =2+(b+1)k+bB$.

The points of $S$ lie on the plane as follows (see \cref{fig:hardness-titto-a}):
\begin{itemize}
    \item $p_1$ is the origin, with coordinates $(0,0)$.
    \item Consider $b+1$ upward rays $\rho_1,\dots,\rho_{b+1}$, whose starting point is $p_1$, such that the angles $\alpha_1,\dots,\alpha_{b+1}$ that they respectively form with the $x$-axis satisfy $3\pi/4 > \alpha_1 > \dots > \alpha_{b+1} > \pi/4$. Let $\ell$ be a line intersecting all the rays, with a positive slope smaller than $\pi/4$. For $j=1, \dots, b+1$, place $k$ points $p_{j,1},\dots,p_{j,k}$ (in this order from bottom to top) along $\rho_j$, so that $p_{j,k}$ is on $\ell$ and no two points share the same $y$-coordinate. Observe that $p_{b+1,k}$ is the highest point placed so far.
    \item Place $p_n$ at coordinates $(0, 10\cdot y(p_{b+1,k}))$. 
    \item Finally, for $j=1, \dots, b$, place $B$ points along a non-horizontal segment $s_j$ in such a way that: (i) $s_j$ is entirely contained in the triangle with vertices $p_{j,k}$, $p_{j+1,k}$, and $p_n$, (ii) for any point $p$ on $s_j$, the polygonal line $\overline{p_1p}\cup \overline{pp_n}$ is contained in the region $R_j$ delimited by the polygon $\overline{p_1p_{j,k}}\cup \overline{p_{j,k}p_n}\cup \overline{p_n p_{j+1,k}}\cup \overline{p_{j+1,k} p_1}$, and (iii) no two distinct points on any two segments $s_i$ and $s_j$ share the same $y$-coordinate. 
\end{itemize}
Note that $S$ has $2+(b+1)k+bB=n$ points. 
This reduction is the key ingredient in proving the following theorem.

\begin{figure}[tb!]
    \centering
    \begin{subfigure}{0.48\textwidth}
    \centering
        \includegraphics[width=\textwidth, page=8]{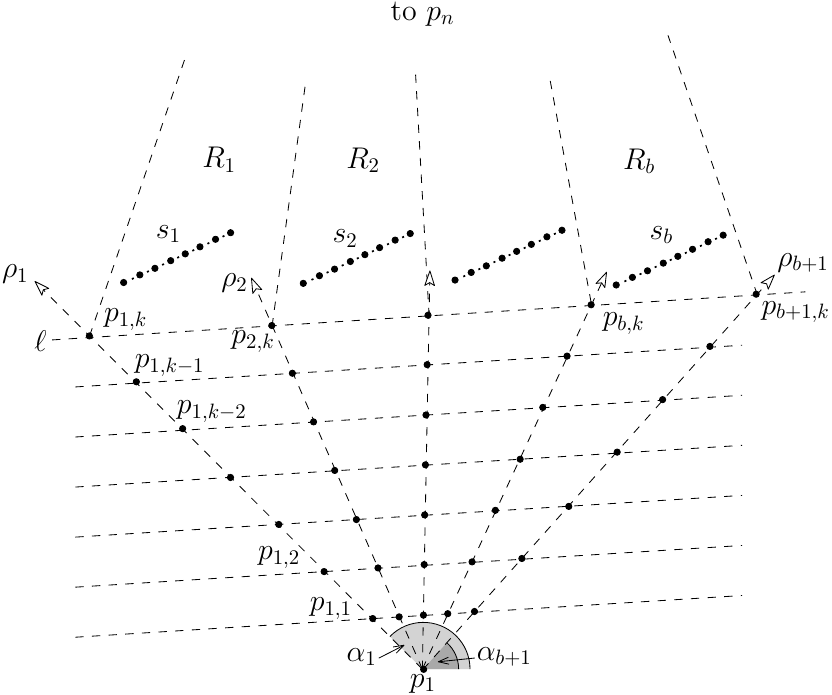}
        \subcaption{}\label{fig:hardness-titto-a}
    \end{subfigure}
    \hfill
    \begin{subfigure}{0.48\textwidth}
     \centering
        \includegraphics[width=\textwidth, page=7]{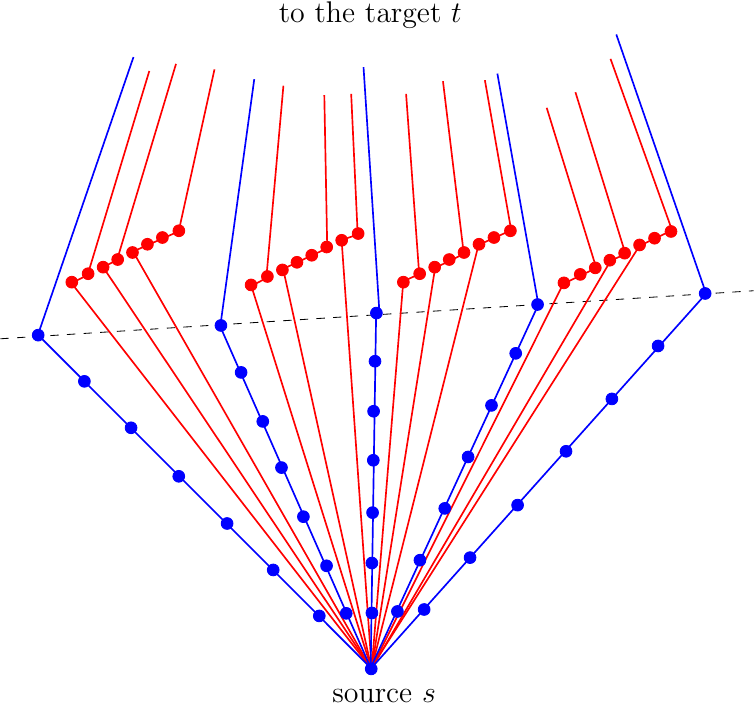}
        \subcaption{}\label{fig:hardness-titto-b}
    \end{subfigure}
    \caption{Illustration for the proof of \cref{th:st-hardness}. (a) The pointset $S$. (b) The UPSE of $G$ on $S$, where the $a_i$-paths are drawn in red and the additional $k$-paths are in blue. The pointset $S$ and the graph $G$ are those resulting from the reduction applied to the instance $A=\{2,2,2,2,2,2,3,3,3,3,4,4\}$.}
    \label{fig:hardness-titto}
\end{figure}

\begin{theorem} \label{th:st-hardness}
{\sc UPSE Testing} is \NP-hard even for planar $st$-graphs consisting of a set of directed internally-disjoint $st$-paths. 
\end{theorem}
\begin{proof}
First, the construction of $G$ and $S$ takes polynomial time. In particular, the coordinates of the points in $S$ can be encoded with a polylogarithmic number of bits. In order to prove the NP-hardness, it remains to show that the constructed instance $(G,S)$ of {\sc UPSE Testing} is equivalent to the given instance $A$ of \textsc{3-Partition}. Refer to \cref{fig:hardness-titto-b}. 

First, suppose that $A$ is a positive instance of \textsc{3-Partition}, that is, there exist sets $A_1,\dots,A_b$, each with three integers, such that the sum of the integers in each set $A_j$ is $B$. We construct an UPSE of $G$ on $S$ as follows. We map $s$ to $p_1$ and $t$ to $p_n$. For $j=1,\dots,b+1$, we map the $k$ internal vertices of a $k$-path to the points $p_{j,1}, \dots p_{j,k}$, so that vertices that come first in the directed path have smaller $y$-coordinates. Furthermore, for $j=1,\dots,b$, let $A_j=\{a_{j_1},a_{j_2},a_{j_3}\}$. Then we map the $a_{j_1}$ internal vertices of an $a_{j_1}$-path, the $a_{j_2}$ internal vertices of an $a_{j_2}$-path, and the $a_{j_3}$ internal vertices of an $a_{j_3}$-path to the set of $B$ points in the triangle with vertices $p_{j,k}$, $p_{j+1,k}$, and $p_n$, so that vertices that come first in the directed paths have smaller $y$-coordinates and so that the internal vertices of the $a_{j_1}$-path have smaller $y$-coordinates than the internal vertices of the $a_{j_2}$-path, which have smaller $y$-coordinates than the internal vertices of the $a_{j_3}$-path. This results in an UPSE  of $G$ on $S$.

Second, suppose that $(G,S)$ is a positive instance of {\sc UPSE Testing}. Trivially, in any UPSE of $G$ on $S$, we have that $s$ is drawn on $p_1$ and $t$ on $p_n$. Consider the points $p_{1,1}, \dots p_{{b+1},1}$. The paths using them use all the $(b+1)k$ points $p_{j,i}$, with $j = 1, \dots, b+1$ and $i=1,\dots,k$. Indeed, if these paths left one of such points unused, no other path could reach it from $s$ without passing through $p_{1,1}, \dots p_{{b+1},1}$, because of the collinearity of the points along the rays $\rho_1,\dots,\rho_{b+1}$. Hence, there are at most $b+1$ paths that use $(b+1)k$ points. Since the $a_i$-paths have less than $k$ internal vertices, these $b+1$ paths must all be $k$-paths. Let $P_1,\dots,P_{b+1}$ be the left-to-right order of the $k$-paths around~$p_1$. For $j=1,\dots,b+1$, path $P_j$ uses all points $p_{j,i}$ on $\rho_j$, as if $P_j$ used a point $p_{h,i}$ with $h>j$, then two among $P_j,\dots,P_{b+1}$ would cross each other. Note that, after using $p_{j,k}$, path $P_j$ ends with the segment $\overline{p_{j,k} p_n}$. Hence, for $j = 1, \dots, b$, the region $R_j$ is bounded by $P_j$ and $P_{j+1}$; recall that $R_j$ contains the segment $s_j$. The $a_i$-paths must then use the points on $s_1,\dots,s_b$. Since $B/4 < a_i < B/2$, no two $a_i$-paths can use all the $B$ points in one region and no four $a_i$-paths can lie in the same region. Hence, three $a_i$-paths use the $B$ points in each region, and this provides a solution to the given \textsc{3-Partition} instance.
\end{proof}

\begin{figure}[tb!]
    \centering
    \begin{subfigure}{0.48\textwidth}
    \centering
        \includegraphics[width=\textwidth, page=9]{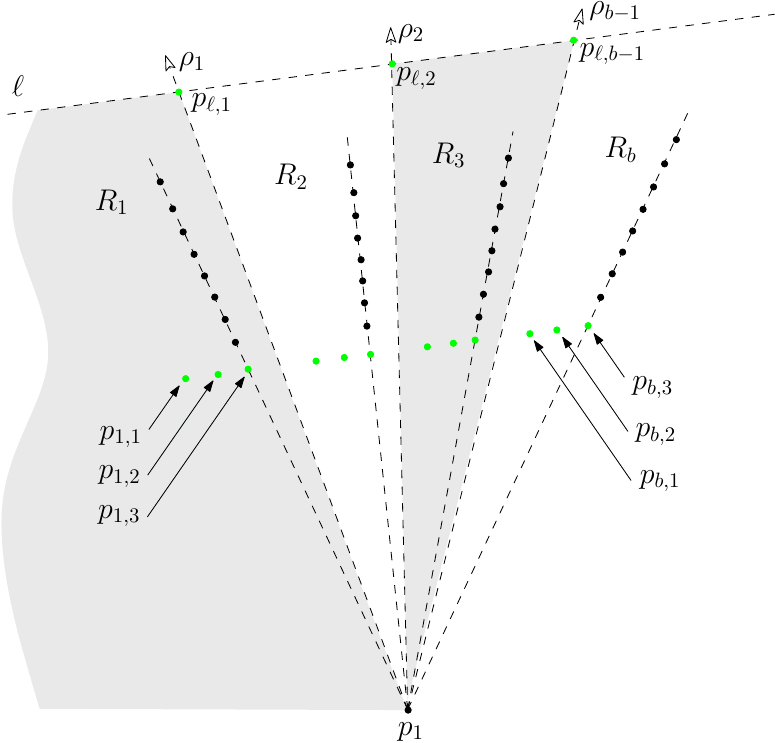}
        \subcaption{}\label{fig:hardness-tree-a}
    \end{subfigure}
    \hfill
    \begin{subfigure}{0.48\textwidth}
     \centering
        \includegraphics[width=\textwidth, page=10]{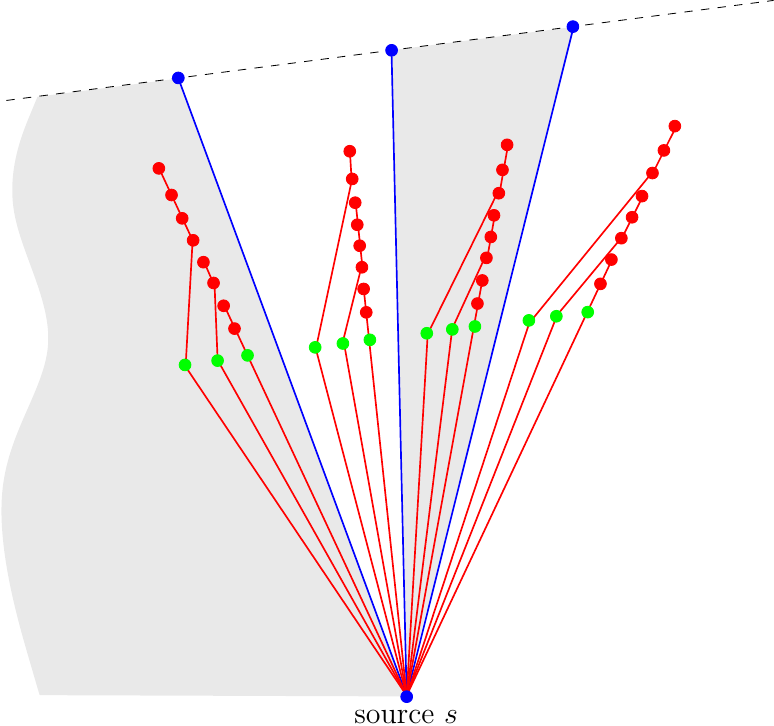}
        \subcaption{}\label{fig:hardness-tree-b}
    \end{subfigure}
    \caption{Illustration for the proof of \cref{th:tree-hardness}. (a) The pointset $S$. The points of $S$ visible from $p_1$ (green points) are as many as the children of the root of the tree $T$. The portions of the regions $R_1,R_2,\dots,R_b$ below the line $\ell$ are alternately colored gray and white.  
    (b) The UPSE of $T$ on~$S$ corresponding to a solution to the original instance  \textsc{3-partition} (red vertices).}
    \label{fig:hardness-tree}
\end{figure}

We next reduce the \textsc{3-Partition} problem to the instances of \textsc{UPSE testing} in which the input is a directed tree consisting of a set of root-to-leaf  paths. Consider an instance of \textsc{3-Partition} consisting of a set $A =\{a_1, \dots, a_{3b}\}$ of $3b$ integers, where $\sum_{i=1}^{3b} a_i = bB$ and $B/4 \leq a_i \leq B/2$, for $i = 1, \dots, 3b$. We construct a directed tree $T$ as follows. The root $s$ of $T$ has $4b-1$ children. Among them, $b-1$ are leaves $v_{1}, \dots, v_{b-1}$, while each of the remaining $3b$ children is the first vertex of a directed path $P_i$, for $i = 1, \dots, 3b$, consisting of the $a_i+1$ vertices $v_{i,1}, v_{i,2}, \dots, v_{i,a_i+1}$, where $v_{i,1}$ is the child of $s$ and $v_{i,a_i+1}$ is a leaf. All the edges of $T$ are directed from the root $s$ to the leaves. Note that the number of vertices of $T$ is $n=1+(b-1)+\sum_{i=1}^{3b}(a_i+1)=b(B+4)$.
The points of $S$ lie on the plane as follows (see \cref{fig:hardness-tree-a}):
\begin{itemize}
    \item $p_1$ is the origin, with coordinates $(0,0)$.
    \item Consider $b-1$ upward rays $\rho_1,\dots,\rho_{b-1}$, whose starting point is $p_1$, such that the angles $\alpha_1,\dots,\alpha_{b-1}$ formed by $\rho_1,\dots,\rho_{b-1}$ with the $x$-axis satisfy $3\pi/4 > \alpha_1 > \dots > \alpha_{b-1} > \pi/4$. These rays split the half plane above the $x$-axis into $b$ regions $R_j$, with $j=1,2, \dots, b$. In the interior of each region $R_j$, place three points $p_{j,1}, p_{j,2},$ and $p_{j,3}$ in such a way that $p_{j,1}$ is lower than $p_{j,2}$, which is lower than $p_{j,3}$, and so that they are all visible from $s$. Along the line passing through $s$ and $p_{j,3}$ place $B$ points above $p_{j,3}$.  
    \item Let $y_{m}$ be the highest $y$-coordinate used so far. 
    Let $\ell$ be a line with positive slope smaller than $\pi/4$ intersecting all the rays $\rho_1,\dots,\rho_{b-1}$ at points that have $y$-coordinates larger than $y_{m}$. For $j=1, \dots, b-1$, place a point $p_{\ell,j}$ at the intersection of $\rho_j$ with~$\ell$.
\end{itemize}

Note that $S$ has $1+3b+bB+(b-1)=b(B+4)=n$ points. This reduction is the key ingredient in proving the following theorem.

\begin{theorem} \label{th:tree-hardness}
{\sc UPSE Testing} is \NP-hard even for directed trees consisting of a set of directed root-to-leaf paths. 
\end{theorem}
\begin{proof}
First, the construction of $T$ and $S$ takes polynomial time. In particular, the coordinates of the points in $S$ can be encoded with a polylogarithmic number of bits. In order to prove the NP-hardness, it remains to show that the constructed instance $(T,S)$ of {\sc UPSE Testing} is equivalent to the given instance $A$ of \textsc{3-Partition}. Refer to \cref{fig:hardness-tree-b}. 

First, suppose that $A$ is a positive instance of \textsc{3-Partition}, that is, there exist sets $A_1,\dots,A_b$, each with three integers, such that the sum of the integers in each set $A_j$ is $B$. We construct an UPSE of $G$ on $S$ as follows. We map $s$ to $p_1$. For $j=1,\dots,b-1$, we map the child $v_j$ of $s$ to $p_{\ell,j}$. Furthermore, for $j=1,\dots,b$, let $A_j=\{a_{j_1},a_{j_2},a_{j_3}\}$. Then we map the $a_{j_1}$ internal vertices of an $a_{j_1}$-path, the $a_{j_2}$ internal vertices of an $a_{j_2}$-path, and the $a_{j_3}$ internal vertices of an $a_{j_3}$-path to the set of $B$ points in the region $R_j$, so that the neighbors of $s$ in the $a_{j_1}$-path, in the $a_{j_2}$-path, and in the $a_{j_3}$-path lie on $p_{j,1}$, $p_{j,2}$, and $p_{j,3}$, respectively, so that vertices that come first in the directed paths have smaller $y$-coordinates, and so that the internal vertices of the $a_{j_1}$-path have larger $y$-coordinates than the internal vertices of the $a_{j_2}$-path, which have larger $y$-coordinates than the internal vertices of the $a_{j_3}$-path. This results in an UPSE of $T$ on $S$.

Second, suppose that $(T,S)$ is a positive instance of {\sc UPSE Testing}. It is obvious that the root $s$ of $T$ has to be placed on $p_1$. From the root $s$ only $4b-1$ points are visible. These are the points $p_{\ell,j}$, for $j=1, \dots, b-1$, and the points $p_{h,1}, p_{h,2},$ and $p_{h,3}$, for $h=1, \dots, b$ (all these points are filled green in \cref{fig:hardness-tree-a}).  
Since $T$ has $4b-1$ children, each child must use one of the above points. Consider point $p_{\ell,b-1}$. Since this is the highest point in the set $S$, the child that uses it must be a leaf. This also holds for $p_{\ell,b-2}$, which is the highest of the remaining points. Iterating this argument we have that the points $p_{\ell,j}$, with $j=1, \dots, b-1$, must be used by the $b-1$ children of $s$ which are leaves of $T$. Since all other vertices have smaller $y$-coordinates, each path $P_i$, with $i=1, \dots, 3m$, is constrained to be into a region $R_j$, with $j=1, \dots, b$ (see \cref{fig:hardness-tree-b}). Since each region $R_j$ contains exactly three points $p_{j,1}, p_{j,2},$ and $p_{j,3}$ visible from $s$, each region hosts exactly three such paths, which use the remaining $B$ points, and this provides a solution to the given \textsc{3-Partition} instance.
\end{proof}

\section[Algorithms for Planar st-Graphs]{UPSE Testing and Enumerating UPSEs for Planar st-Graphs with Maximum st-Cutset of Bounded Size}\label{sec:k-paths}

An \emph{$st$-cutset} of a planar $st$-graph $G = (V,E)$ is a subset $W$ of $E$ such that: 
\begin{itemize}
    \item removing $W$ from $E$ results in a graph consisting of exactly two connected components $C_s$ and $C_t$,
    \item $s$ belongs to $C_s$ and $t$ belongs to $C_t$, and
    \item any edge in $W$ has its tail in $C_s$ and its head in $C_t$. 
\end{itemize}
In this section, we consider instances $(G,S)$ where $G$ is a planar $st$-graph, whose maximum $st$-cutset has bounded size $k$. In \cref{th:k-paths}, we show that {\sc UPSE Testing} can be solved in polynomial time for such instances $(G,S)$. Moreover, in \cref{th:st-enumeration}, we show how to enumerate all UPSEs of $(G,S)$ with linear delay. 
The algorithm for \cref{th:k-paths} is based on a dynamic programming approach. 
It exploits the property that, for an $st$-cutset $W$ defining the connected components $C_s$ and $C_t$, the extensibility of an UPSE $\Gamma'$ of $C_s \cup W$ on a subset $S'$ of $S$ to an UPSE of $G$ on $S$ only depends on the drawing of the edges of $W$, and not on the embedding of the remaining vertices~of~$C_s$, provided that in $\Gamma'$ there exists an horizontal line that crosses all the edges of $W$.
The algorithm for \cref{th:st-enumeration} leverages a variation of the dynamic programming table computed by the former algorithm to efficiently test the extensibility of an UPSE of $C_s \cup W$ (in which there exists a horizontal line that crosses all the edges of $W$) on a subset $S'$ of $S$ to an UPSE of $G$ on $S$.

The proofs of \cref{th:k-paths,th:st-enumeration} exploit two dynamic programming tables $T$ and $Q$ defined as follows. Each entry of $T$ and $Q$ is indexed by a {\em key} that consists of a set of $h \leq k$ triplets $\langle e_i,p_i,q_i\rangle$, where, for any $i = 1,\dots,h$, it holds that $e_i \in E(G)$, $p_i,q_i \in S$, and $y(p_i) < y(q_i)$.
Moreover, each key $\chi = \bigcup^h_{i=1} \langle e_i,p_i,q_i\rangle$ satisfies the following constraints:
\begin{itemize}
\item the set $E(\chi) =\bigcup^h_{i=1} e_i$ is an $st$-cutset of $G$ and, for every $i,j$, with $i \neq j$, it holds true that $e_i \neq e_j$ (that is, $|E(\chi)| = h$);
\item for every $i,j$, with $i \neq j$, it holds true that $p_i = p_j$ (resp.\ that $q_i = q_j$) if and only if $e_i$ and $e_j$ have the same tail (resp.\ the same head); and
\item let $\ell_\chi$ be the horizontal line passing through the tail with largest $y$-coordinate among the edges in  $E(\chi)$, i.e., $\ell_\chi:=y=y(p_i)$ such that $y(p_j)\leq y(p_i)$ for any $\langle e_j,p_j,q_j\rangle \in \chi$; then $\ell_\chi$ intersects all the segments $\overline{p_j q_j}$, possibly at an endpoint.
\end{itemize}
For brevity, we sometimes say that the edge $e_i$ has its tail (resp.\ its head) \emph{mapped by} $\chi$ on $p_i$ (resp.\ on $q_i$). We also say that $e_i$ is \emph{drawn as in} $\chi$ if its drawing is the segment $\overline{p_i q_i}$.

Let $\chi = \bigcup^h_{i=1} \langle e_i, p_i, q_i\rangle$ be a key of $T$ and of $Q$; see \cref{fig:ST-digraphs-inductive-chi}. Let $G_{\chi}$ be the connected component containing $s$ of the graph obtained from $G$ by removing the edge set $E(\chi)$.

The entry $\Tparam{\chi}$ contains a Boolean value such that $\Tparam{\chi} = \texttt{True}$ if and only if there exists an \UPSE of $G_{\chi}^{+}=G_{\chi} \cup E(\chi)$ on some subset $S' \subset S$ with $|S'| = |V(G_{\chi}^{+})|$ such that:
\begin{itemize}
\item the lowest point $p_s$ of $S$ belongs to $S'$ and $s$ lies on it, and 
\item for $i=1,\dots,h$, the edge $e_i$ is drawn as in $\chi$.
\end{itemize}
If $\Tparam{\chi}=\texttt{False}$, the entry $\Qparam{\chi}$ contains the empty set $\varnothing$. If $\Tparam{\chi}=\texttt{True}$ and $E(\chi)$ coincides with the set of edges incident to $s$, then $\Qparam{\chi}$ stores the set $\{\bot\}$. If $\Tparam{\chi}=\texttt{True}$ and $E(\chi)$ does not coincide with the set of edges incident to $s$, $\Qparam{\chi}$ stores the set $\Phi$ of keys with the following properties. Let $e_\tau$ be any edge whose tail $v_\tau$ has maximum $y$-coordinate among the edges in $E(\chi)$, i.e., $\langle e_\tau,p_\tau,q_\tau\rangle$ is such that $y(p_\tau)\geq y(p_j)$ for any $\langle e_j,p_j,q_j\rangle \in \chi$.
For each $\varphi \in \Phi$, we have that:
\begin{itemize}
\item $\Tparam{\varphi} = \texttt{True}$;
\item $E(\chi) \cap E(\varphi)$ contains all and only the edges in $E(\chi)$ whose tail is not $v_\tau$, and each edge $e_i \in E(\chi) \cap E(\varphi)$ is drawn in $\varphi$ as it is drawn in $\chi$; and
\item all the edges in $E(\varphi) \setminus E(\chi)$ have $v_\tau$ as their head.
\end{itemize}

\begin{figure}[b!]
    \centering
    \begin{subfigure}[b]{0.45\textwidth}
    \centering
        \includegraphics[width=.7\textwidth, page=1]{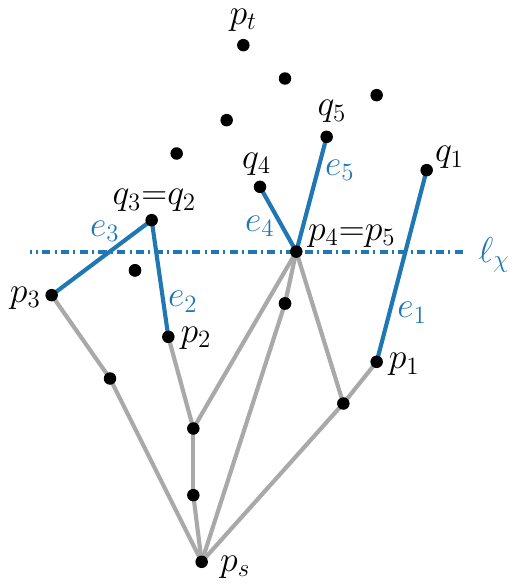}
        \subcaption{}
        \label{fig:ST-digraphs-inductive-chi}
    \end{subfigure}
    \hfill
    \begin{subfigure}[b]{0.45\textwidth}
    \centering
        \includegraphics[width=.7\textwidth, page=3]{img/ST-digraphs-inductive-case.pdf}
        \subcaption{}
        \label{fig:ST-digraphs-inductive-phi-true}
    \end{subfigure}
    \caption{(a) An entry $\chi = \bigcup^5_{i=1} \langle e_i,p_i,q_i \rangle$ with $\Tparam{\chi} = \texttt{True}$ and a corresponding UPSE of $G_\chi$ on a subset of $S$ that includes $p_s$. The edges in $E(\chi)$ are colored blue.     (b) An entry $\varphi$ from which $\chi$ stems; the points in $S_\downarrow$ are filled white. The edges in $H^-$ are colored green, while the edges in $H^+$ are colored orange.}
    \label{fig:ST-digraphs-inductive}
\end{figure}

Additionally, we store a list $\Lambda$ of the keys $\sigma$ such that $\Tparam{\sigma} = \texttt{True}$ and $E(\sigma)$ is the set of edges incident to $t$. Note that each edge in  $E(\sigma)$ has its head mapped by $\sigma$ to the point $p_t \in S$ with largest $y$-coordinate.

We use dynamic programming to compute the entries of $T$ and $Q$ in increasing order of $|V(G_{\chi})|$.
By the definition of $T$, we have that $G$ admits an \UPSE on $S$ if and only if $\Lambda \neq \varnothing$.

First, we initialize all entries of $T$ to $\texttt{False}$ and all entries of $Q$ to $\varnothing$.

If $|V(G_{\chi})|=1$, then $G_{\chi}$ only consists of $s$. We set $\Tparam{\chi} = \texttt{True}$ and $\Qparam{\chi}=\{\bot\}$ for every key $\chi = \bigcup^h_{i=1} \langle e_i,p_i,q_i\rangle$ such that:
\begin{itemize}
\item $e_1,\dots,e_h$ are the edges incident to $s$;
\item $p_1=\dots=p_h=p_s$; and
\item for every distinct $i$ and $j$ in $\{1,\dots,h\}$, we have that $p_s$, $q_i$, and $q_j$ are not aligned.
\end{itemize}

If $|V(G_{\chi})|>1$, we compute $\Tparam{\chi}$ and $\Qparam{\chi}$ as follows, see \cref{fig:ST-digraphs-inductive-phi-true}. If two segments $\overline{p_i q_i}$ and $\overline{p_j q_j}$, with $i \neq j$, cross (that is, they share a point that is internal for at least one of the segments), then we leave $\Tparam{\chi}$ and $\Qparam{\chi}$ unchanged; in particular, $\Tparam{\chi} = \texttt{False}$ and $\Qparam{\chi} = \varnothing$.
Otherwise, we proceed as follows. 
 Let $e_\tau$ be any edge whose tail $v_\tau$ has maximum $y$-coordinate among the edges in $E(\chi)$. Let $H^-$ be the set of edges obtained from $E(\chi)$ by removing all the edges having $v_\tau$ as their tail, and let $H^+$ be the set of edges of $G$ having $v_\tau$ as their head.
We define the set $H := H^- \cup H^+$. We have the following.

\begin{figure}[tb]
    \centering
    \begin{subfigure}[b]{0.3\textwidth}
    \centering
        \includegraphics[width=\textwidth, page=4]{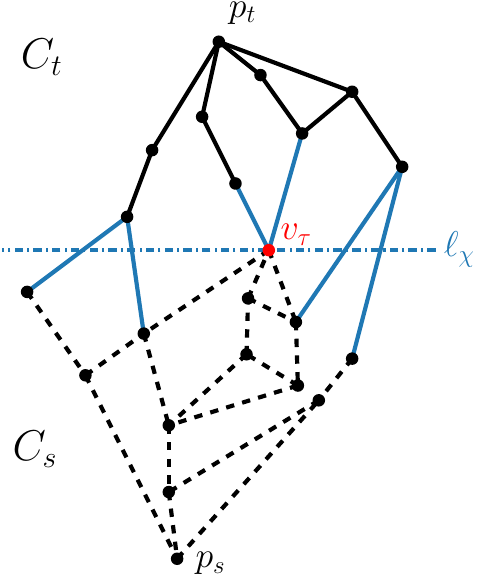}
        \subcaption{}
        \label{fig:ST-digraphs-inductive-cutset-B}
    \end{subfigure}
    \hspace{0.1\textwidth}
    \begin{subfigure}[b]{0.3\textwidth}
    \centering
        \includegraphics[width=\textwidth, page=5]{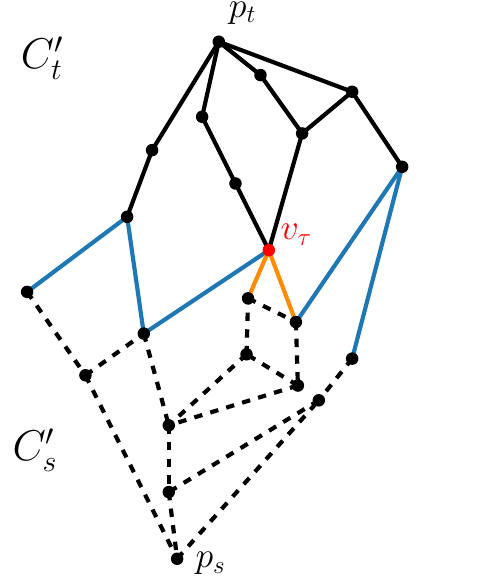}
        \subcaption{}
        \label{fig:ST-digraphs-inductive-cutset-B-minus}
    \end{subfigure}
    \caption{Illustrations for \cref{cl:st-cutset}. (a) The connected components $C_s$ (dashed) and $C_t$ (solid black) defined by the $st$-cuteset $E(\chi)$. (b) The connected components $C'_s$ (dashed) and $C'_t$ (solid black) defined by the $st$-cuteset $H$ (blue and orange edges).
    }
    \label{fig:ST-digraphs-inductive-cutset}
\end{figure}

\begin{claimx}\label[claimx]{cl:st-cutset}
$H$ is an $st$-cutset of $G$.
\end{claimx}

\begin{proof}
Recall that, since $E(\chi)$ is an $st$-cutset, removing the edges of $E(\chi)$ from $G$ yields two connected components $C_s$ and $C_t$ such that $s$ belongs to $C_s$ and $t$ belongs to $C_t$; see \cref{fig:ST-digraphs-inductive-cutset-B}.
Let $C'_t$ be the graph consisting of $C_t$,  the vertex $v_\tau$, and the edges having $v_\tau$ as their tail (these are the edges in $E(\chi) \setminus H^-$, which are not part of $H$).
Also, let $C'_s$ be the graph obtained by removing from $C_s$ the vertex $v_\tau$ and the edges in $H^+$ (i.e., these are the edges outgoing from $v_\tau$); see \cref{fig:ST-digraphs-inductive-cutset-B-minus}.
We have that $C'_s$ and $C'_t$ do not share any vertex, since $C_s$ and $C_t$ do not share any vertex, since $V(C'_s)\subset V(C_s)$ and since the only vertex in $V(C'_t)\setminus V(C_t)$ is $v_\tau$, which does not belong to $C'_s$. Moreover, by construction $G = C'_t \cup C'_s \cup H$, in particular the only edges connecting vertices in $C'_s$ with vertices in $C'_s$ are those in $H$, which have their tails in $C'_s$ and their heads in $C'_t$. Also, we have that $s$ belongs to $C'_s$ and $t$ belongs to $C'_t$. To prove that $H$ is an $st$-cutset of~$G$, it only remains to argue that each of $C'_s$ and $C'_t$ is connected.
Since $C_t \subseteq C'_t$ and since $C_t$ is connected, we have that every pair of vertices distinct from $v_\tau$ is connected by an undirected path in $C'_t$. Also, the heads of the edges outgoing from $v_\tau$ belong to $C_t$ and, by construction, such edges belong to $C'_t$. Hence, there exists an undirected path in $C'_t$ between $v_\tau$ and every vertex of $C_t$. Therefore, $C'_t$ is connected.
Now, suppose, for a contradiction, that $C'_s$ is not connected and thus there exists a vertex $v$ which is not in the same connected component as $s$ in $C'_s$. Since $G$ is a planar $st$-graph, it contains a directed path from $s$ to $v$. If such a path does not belong entirely to $C'_s$, it contains an edge which is directed from a vertex not in $C'_s$ to a vertex in $C'_s$. Moreover, such an edge belongs to $H$, however we already observed that all the edges in $H$ are outgoing from the vertices in $C'_s$, a contradiction.
\end{proof}

Consider the set $S_\downarrow$ consisting of the points in $S$ whose $y$-coordinates are smaller than~$y(p_\tau)$. We have the following crucial observation.

\begin{observation}\label[observation]{claim:extension}
$\Tparam{\chi} = \texttt{True}$ if and only if there exists some key $\varphi$, with $E(\varphi) = H$, such that $\Tparam{\varphi}=\texttt{True}$, the edges in $H^-$ are drawn in $\varphi$ as in $\chi$, the edges in $H^+$ have their heads mapped by $\varphi$ on $p_\tau$ and their tails on a point in $S_\downarrow$. 
\end{observation}

In view of \cref{claim:extension}, we can now define a procedure to compute $\Tparam{\chi}$ and $\Qparam{\chi}$. Assume that the edges $e_1,\dots,e_{|H^-|},\dots,e_{|H|} \in H$ are ordered so that the edges of $H^-$ precede those of $H^+$. By \cref{claim:extension}, if $|S_\downarrow| < |H^+|$, then we leave $\Tparam{\chi}$ and $\Qparam{\chi}$ unchanged, i.e., $\Tparam{\chi} = \texttt{False}$ and $\Qparam{\chi} = \varnothing$.
In fact, in this case, there are not enough points in $S_\downarrow$ to map the tails of the edges in $H^+$.
Otherwise, let $D$ be the set of all permutations with repetitions of $|H^+|$ points from $S_\downarrow$. We define a set $\Phi$ of keys that, for each $(d_1,\dots,d_{|H^+|})\in D$, contains a key $\varphi$ such that:
\begin{enumerate}
    \renewcommand{\theenumi}{\roman{enumi}} %
    \renewcommand{\labelenumi}{\textbf{(\theenumi)}} %
\item $E(\varphi) = H$; 
\item for any $i=1,\dots,|H^-|$, the triple containing $e_i$ in $\varphi$ is the same as the triple containing $e_i$ in $\chi$ (note that $e_i \in H^-$);
\item for any $j= |H^-|+1,\dots,|H|$, the triple containing $e_j$ in $\varphi$ has $q_j = p_\tau$, and $p_j=d_{j-|H^-|}$ (note that $e_j \in H^+$); and
\item for every $i=1,\dots,|H^-|$ and $j= |H^-|+1,\dots,|H|$, it holds $p_i = p_j$ if and only if $e_i$ and $e_j$ have the same tail.
\end{enumerate}
Let $\Phi^{\texttt{T}} = \{\varphi: \varphi \in \Phi \wedge \ \Tparam{\varphi} = \texttt{True}\}$.
By \cref{claim:extension}, we have $\Tparam{\chi} = \texttt{True}$ if and only if $|\Phi^\texttt{T}|\geq 1$. Thus, we set $\Tparam{\chi}=\bigvee_{\varphi\in\Phi} \Tparam{\varphi}$ and $\Qparam{\chi} = \Phi^\texttt{T}$.
We say that $\chi$ \emph{stems from} any key $\varphi \in \Phi$ with $\Tparam{\varphi}=\texttt{True}$.

We now upper bound the sizes of $T$ and $Q$ and the time needed to compute them.

\begin{claimx}\label[claimx]{cl:table-space}
Tables $T$ and $Q$ have size in $\bigO(n^{3k})$ and $\bigO(kn^{4k}\log n)$, respectively.
\end{claimx}

\begin{proof}
First, we give an upper bound on the number of entries of $T$ (and, thus, of $Q$), which we denote by $\rho$.
Each entry of $T$ is associated with a key $\chi$ defined by an $st$-cutset $E(\chi)$ of size at most $h \leq k$, a permutation (possibly with repetitions) of $h$ points in $S$ describing a mapping of the tails of the edges in $E(\chi)$ with points in $S$, and a permutation (possibly with repetitions) of $h$ points in $S$ describing a mapping of the heads of the edges in $E(\chi)$ with points in $S$. 
Recall that $|S|=n$, that $\binom{a}{b} \leq a^b$, and that the number of permutations with repetition of $h$ elements from a set $U$ is $|U|^h$.
Therefore, we have that $\rho \leq  \binom{m}{k} \cdot n^{k} \cdot n^{k} \leq (mn^2)^k$. Since $m \in \bigO(n)$, we thus have $\rho \in \bigO(n^{3k})$.

We can now upper bound the size of $T$ and $Q$.
Since each entry of $T$ stores a single bit, we immediately have that $T$ has $\bigO(n^{3k})$ size.
Instead, each entry of table $Q$ stores $\bigO(n^k)$ keys of size $\bigO(k \log n)$; thus, $Q$ has $\bigO(kn^{4k} \log n)$ size. 
The upper bound on the number of keys comes from the number of ways to map the tails of the at most $k$ edges incoming into $v_\tau$ on the points of $S$, which has size $n$; this number is $\binom{n}{k} \in \bigO(n^k)$. The upper bound on the size of each key comes from the fact that it consists of at most $k$ triplets each containing an identifier of $\bigO(n)$ edges and two identifiers of $\bigO(n)$ points. 
\end{proof}

\begin{claimx}\label[claimx]{cl:table-time}
Tables $T$ and $Q$ can be computed in $\bigO(n^{4k})$ and $\bigO(kn^{4k}\log n)$ time, respectively.
\end{claimx}

\begin{proof}
We determine the time needed to compute, for each key $\chi$, the value $\Tparam{\chi}$ and $\Qparam{\chi}$.
For each key $\chi$, we need to verify whether the $h$ segments $\overline{p_i q_i}$ intersect at a point different from a common endpoint, which can be tested in $\bigO(k \log k)$ time~\cite{shamos1976geometric}. Moreover, if $|V(G_\chi)|>1$, computing $\Tparam{\chi}$ requires accessing the value of up to $|S_\downarrow|^{|H^+|} < n^k$ entries of $T$, and verifying whether at least one of them contains the value $\texttt{True}$. Since $n>k$, the $\bigO(k \log k)$ term in the running time is dominated by the $\bigO(n^k)$ term, hence the time needed to compute each entry $\Tparam{\chi}$ is thus $\bigO(n^k)$. 
Since, by \cref{cl:table-space}, there are $\bigO(n^{3k})$ keys $\chi$, it follows that $T$ can be computed in overall $\bigO(n^{4k})$ time. On the other hand, the time needed to compute each entry $\Qparam{\chi}$ is upper bounded by the time needed to write the $\bigO(n^k)$ keys contained in $\Qparam{\chi}$, each of which has $\bigO(k \log n)$ size, i.e., $\bigO(kn^k\log{n})$ time per entry.
It follows that $Q$ can be computed in overall $\bigO(kn^{4k}\log{n})$ time.
\end{proof}

Finally, in order to verify whether $G$ admits an \UPSE on $S$, we need to check whether $\Lambda\neq \varnothing$.
Computing the maximum size of an $st$-cutset of a planar $st$-graph $G$ can be done in linear time, as it reduces to the problem of computing the length of a shortest path in the dual of any embedding of $G$ (between the vertices representing the left and right outer faces of this embedding)~\cite{0023376,DBLP:journals/tcs/BattistaT88}.
Therefore, the overall running time to test whether $G$ admits an \UPSE on $S$ is dominated by the time needed to compute $T$, that is, $\bigO(n^{4k})$ time. 

To obtain an \UPSE $\Gamma$ of $G$ on $S$, if any, we proceed as follows.
Suppose that the algorithm terminates with a positive answer and let $\sigma$ be any key in $\Lambda$. 
We initialize $\Gamma$ to a drawing of the edges in $E(\sigma)$, where each edge $e_i \in E(\sigma)$ is drawn as in $\sigma$. 
Then, in $\bigO(n^k)$ time, we can search in $T$ a key $\chi$ with $\Tparam{\chi} = \texttt{True}$ such that $\sigma$ stems from $\chi$, and update $\Gamma$ to include a drawing of the edges in $E(\chi)\setminus  E(\sigma)$, where each edge $e_i \in E(\chi)\setminus  E(\sigma)$ is drawn as in $\chi$; note that the edges in $E(\chi) \cap E(\sigma)$ are drawn in $\chi$ as they are drawn in $\sigma$. Applying such a procedure until a key $\alpha$ is reached such that $\Tparam{\alpha}=\texttt{True}$ and $E(\alpha)$ is the set of edges incident to $s$ yields the desired UPSE of $G$ on $S$. Note that the tail with largest $y$-coordinate among the edges in $E(\sigma)$ is higher than the horizontal line through the tail with largest $y$-coordinate among the edges in $E(\chi)$, hence the depth of the recursion is linear in the size of $G$. We can therefore compute $\Gamma$ in $\bigO(n^{k+1})$ time.

From the above discussion, we have the following theorem.

\begin{theorem} \label{th:k-paths}
Let $G$ be an $n$-vertex planar $st$-graph whose maximum $st$-cutset has size~$k$ and let $S$ be a set of $n$ points. {\sc UPSE Testing} can be solved for $(G,S)$ in $\bigO{(n^{4k})}$ time and $\bigO{(n^{3k})}$ space; if an \UPSE of $G$ on $S$ exists, it can be constructed within the same~bounds.
\end{theorem}

We now turn our attention to the design of an algorithm for the enumeration of the UPSEs of $G$ on $S$. The algorithm exploits the table $Q$ and the set $\Lambda$. By \cref{cl:table-space,cl:table-time}, these can be computed in $\bigO(k n^{4k} \log{n})$ time~and~space. Our enumeration algorithm defines and explores an acyclic digraph $\mathcal{D}$. The nodes of the digraph correspond to the keys $\chi$ of the dynamic programming table $Q$ such that $\Qparam{\chi}\neq \varnothing$, plus a source $n_{\cal S}$ and a sink $n_{\cal T}$. Let $\chi_i$ and $\chi_j$ be two keys of $Q$ such that $\Qparam{\chi_i}\neq\varnothing$ and $\Qparam{\chi_j}\neq \varnothing$, and let $n(\chi_i)$ and $n(\chi_j)$ be the nodes corresponding to $\chi_i$ and $\chi_j$ in $\mathcal{D}$, respectively. There exists an edge directed from $n(\chi_i)$ to $n(\chi_j)$ in $\mathcal{D}$ if $\chi_j \in \Qparam{\chi_i}$. Also, there exists an edge directed from $n_{\cal S}$ to each node $n(\sigma)$ such that $\sigma \in \Lambda$. Finally, there exists an edge directed to $n_{\cal T}$ from each node $n(\chi)$ such that $\Qparam{\chi}=\{\bot\}$.
Note that $n_{\cal S}$ is the unique source of $\mathcal{D}$, $n_{\cal T}$ is the unique sink of $\mathcal{D}$, and $\mathcal{D}$ has no directed cycle.
Hence, $\mathcal{D}$ is an $n_{\cal S}n_{\cal T}$-graph.

The exploration of $\mathcal{D}$ performed by our enumeration algorithm is a depth-first traversal. Every distinct path in $\mathcal{D}$ from $n_{\cal S}$ to $n_{\cal T}$ corresponds to an UPSE of $G$ on $S$. We initialize a current UPSE $\Gamma$ of $G$ on $S$ as $\Gamma=S$ (where no edge of $G$ is drawn). When the visit traverses an edge of $\mathcal{D}$ directed from a node $n(\chi_i)$ to a node $n(\chi_j)$, it adds to $\Gamma$ the edges in $E(\chi_j)\setminus E(\chi_i)$, drawn as in $\chi_j$. Note that these are all the edges in $E(\chi_j)$ if $n(\chi_i)=n_{\cal S}$ and it is an empty set if $n(\chi_j)=n_{\cal T}$. Whenever the traversal reaches $n_{\cal T}$, it outputs the constructed UPSE $\Gamma$ of $G$ on $S$. When the visit backtracks on a node $n(\chi_i)$ coming from an edge $(n(\chi_i), n(\chi_j))$, it removes from $\Gamma$ the edges in $E(\chi_j)\setminus E(\chi_i)$. 

To prove the correctness of the algorithm, we show what follows: \begin{enumerate}
    \renewcommand{\theenumi}{\roman{enumi}} 
    \renewcommand{\labelenumi}{\textbf{(\theenumi)}} 
    \item \label{prop:distinct-drawing} Distinct $n_{\cal S}n_{\cal T}$-paths in $\mathcal D$ correspond to different UPSEs of $G$ on $S$.
    \item \label{prop:all-drawings} For each UPSE of $G$ on $S$, there exists in $\mathcal D$ an $n_{\cal S}n_{\cal T}$-path corresponding to it.
\end{enumerate}
For a directed path $\mathcal{P}$ in $\mathcal D$, let $E(\mathcal {P})$ be the set that contains all the edges in the sets $E(\chi)$, where $\chi$ is any key corresponding to a node in $\mathcal{P}$. 

\begin{itemize}
    \item To prove \cref{prop:distinct-drawing}, we proceed by contradiction. Let $\Gamma$ be an UPSE of $G$ on $S$ that is generated twice by the algorithm, when traversing distinct $n_{\cal S}n_{\cal T}$-paths $\mathcal{P}_1$ and $\mathcal{P}_2$. Let $n(\chi_x)$ be the closest node to $n_\mathcal{S}$ in $\mathcal{P}_1$ and $\mathcal{P}_2$ such that $(n(\chi_x),n(\chi_1))$ is an edge in $\mathcal{P}_1$ and  $(n(\chi_x),n(\chi_2))$ is an edge in $\mathcal{P}_2$, with $n(\chi_1) \neq n(\chi_2)$, where $\chi_x$, $\chi_1$, and $\chi_2$ are keys of $Q$. Note that, since the path $\mathcal{P}_x$ from $n_\mathcal{S}$ to $n(\chi_x)$ (possibly such a path is a single node if $n_\mathcal{S}=n(\chi_x)$) is the same in $\mathcal{P}_1$ and $\mathcal{P}_2$, the restriction $\Gamma_x$ of $\Gamma$ to the edge set $E(\mathcal {P}_x)$ is also the same in $\mathcal{P}_1$ and $\mathcal{P}_2$. Hence, the tail $v_{\chi_x}$ with largest $y$-coordinate of an edge in $E(\chi_x)$ is uniquely defined by $\Gamma_x$. This implies that the edge sets $E(\chi_1)$ and $E(\chi_2)$ coincide, as they are both obtained from $E(\chi_x)$ by replacing the edges outgoing from $v_{\chi_x}$ with the edges incoming into $v_{\chi_x}$ in $G$. Since $E(\chi_1)= E(\chi_2)$ and $\chi_1\neq\chi_2$, it follows that $\chi_1$ and $\chi_2$ must differ in the way such keys map the tails of the edges incoming into $v_{\chi_x}$ to the points of $S$. Then the UPSEs yielded by $\mathcal{P}_1$ and $\mathcal{P}_2$ are different, a contradiction.

    \item To prove \cref{prop:all-drawings}, we show that, if there exists an UPSE $\Gamma$ of $G$ on $S$, then there exists a path in $\mathcal D$ from $n_\mathcal{S}$ to $n_\mathcal{T}$ that yields $\Gamma$. For $i=1,\dots,n$, let $S_i$ be the set that consists of the lowest $i$ points of $S$. Also, for $i=1,\dots,n-1$, let $\Gamma_i$ be the restriction of $\Gamma$ to the vertices of $G$ mapped to $S_i$ and to all their incident edges, including those whose other end-vertex is not in $S_i$. We claim that there exists a path $\mathcal{P}_i$ in $\mathcal{D}$ that starts from a node $n_i$ and ends at $n_\mathcal{T}$  such that: (1) the set $E(\mathcal{P}_i)$ coincides with the set of edges that are embedded in $\Gamma_i$; (2) the embedding of the edges in $E(\mathcal{P}_i)$ defined by the keys $\chi$ corresponding to nodes in $\mathcal{P}_i$ is the same as in $\Gamma_i$; and (3) let $\chi_i$ be the key corresponding to $n_i$, then $E(\chi_i)$ contains all and only the edges $e$ of $G$ such that an end-vertex of $e$ is mapped by $\Gamma$ to a point in $S_i$ and the other end-vertex of $e$ is mapped by $\Gamma$ to a point not in $S_i$. The claim implies \cref{prop:all-drawings}, as when $i=n-1$, we have that $E(\mathcal{P}_{n-1})$ is the edge set of $G$, by (1), and that the embedding of the edges in $E(\mathcal{P}_{n-1})$ defined by the keys $\chi$ corresponding to nodes in $\mathcal{P}_{n-1}$ is $\Gamma$, by (2), hence $(n_\mathcal{S},\chi_{n-1})\cup \mathcal{P}_{n-1}$ is the desired path from $n_\mathcal{S}$ to $n_\mathcal{T}$ that yields $\Gamma$.
    
    In order to prove the claim, we proceed by induction. In the base case, we have $i=1$, hence $S_1$ consists only of the point $p_s$ and $\Gamma_1$ is the restriction of $\Gamma$ to all the edges incident to $s$. Since $\Gamma$ is an UPSE, $\Gamma_1$ is an embedding of such edges in which $s$ lies on $p_s$ and any two heads of such edges are not aligned with $p_s$. Hence, by construction, there is a key $\chi$ such that  $E(\chi)$ consists of the set of edges incident to $s$, such that $\Qparam{\chi}=\{\bot\}$, and such that the embedding of the edges in $E(\chi)$ on $S$ defined by $\chi$ is $\Gamma_1$. It follows that $\mathcal D$ contains a node $n(\chi)$ corresponding to $\chi$, and thus a path $\mathcal{P}_1=(n(\chi),n_\mathcal{T})$ with the properties required by the claim.

    For the inductive case, we have $i>1$. Let $p_i$ be the point of $S_i$ with highest $y$-coordinate and let $v_i$ be the vertex of $G$ mapped to $p_i$ by $\Gamma$. By induction, there exists a path $\mathcal{P}_{i-1}$ in $\mathcal{D}$ that starts from a node $n_{i-1}$ and ends at $n_\mathcal{T}$  such that: (1) the set $E(\mathcal{P}_{i-1})$ is the set of edges embedded in $\Gamma_{i-1}$; (2) the embedding of the edges in $E(\mathcal{P}_{i-1})$ defined by the keys corresponding to nodes in $\mathcal{P}_{i-1}$ defines $\Gamma_{i-1}$; and (3) let $\chi_{i-1}$ be the key corresponding to $n_{i-1}$, then $E(\chi_{i-1})$ contains all and only the edges whose end-vertices are mapped by $\Gamma$ one  to a point in $S_{i-1}$ and the other  to a point not in $S_{i-1}$. Note that (3) ensures that all the edges incoming into $v_i$ are in $E(\chi_{i-1})$. 

    Consider the edge set $H_i$ composed of the edges outgoing from $v_i$ and of the edges in $E(\chi_{i-1})$, except for those incoming into $v_i$. We prove that $H_i$ is an $st$-cutset. Indeed, by (3), every edge of $G$ that in $\Gamma$ starts from a point below $p_i$ and ends at a point above $p_i$ is in $E(\chi_{i-1})$. Then $H_i$ comprises all the edges that start from $p_i$ or from a point below $p_i$ and end at a point above $p_i$. Hence, the removal of the edges of $H_i$ splits $G$ into two connected subgraphs, one induced by the vertices (including $s$) mapped by $\Gamma$ to $S_i$, and one induced by the vertices (including $t$) mapped by $\Gamma$ to the points above $p_i$.

    Since $H_i$ is an $st$-cutset, there exists a key $\chi_i$ such that $E(\chi_i)=H_i$ and the edges of $E(\chi_i)$ are embedded in $\chi_i$ as in $\Gamma_i$. Note that $p_i$ is the tail of an edge in $E(\chi_i)$ with largest $y$-coordinate, hence our algorithm, starting from the $st$-cutset $E(\chi_i)$, removes the edges outgoing from $v_i$, and adds the edges incoming into $v_i$, thus it constructs the $st$-cutset $E(\chi_{i-1})$ and, from there, the key $\chi_{i-1}$ in which the edges of $E(\chi_{i-1})$ are mapped as in $\Gamma_{i-1}$. The algorithm then inserts $\chi_{i-1}$ in $\Qparam{\chi_i}$, and hence the digraph $\mathcal D$ contains the edge $(n_i,n_{i-1})$, where $n_i$ is the node of $\mathcal D$ corresponding to $\chi_i$. This completes the induction, hence the proof of the claim and the one of \cref{prop:all-drawings}.
\end{itemize}

It remains to discuss the running time of our enumeration algorithm. Concerning the set-up time, the table $Q$  can be constructed in $\bigO(kn^{4k}\log n)$ time, by \cref{cl:table-time}. Also, the digraph $\mathcal D$ can be constructed in linear time in the size of $Q$, which is $\bigO(kn^{4k}\log n)$ by \cref{cl:table-space}; indeed, the edges outgoing from a node $n(\chi)$ in $\mathcal D$ are those toward the nodes whose corresponding keys are in $\Qparam{\chi}$. Concerning the space usage, again by \cref{cl:table-space}, we have that $Q$ and $\mathcal D$ have $\bigO(kn^{4k}\log n)$ size. Finally, we discuss the delay of our algorithm. The paths from $n_{\mathcal S}$ to $n_{\mathcal T}$ have $\bigO(n)$ size; indeed, each edge $(n(\chi),n(\chi'))$ is such that the horizontal line through the tail with largest $y$-coordinate among the edges in $E(\chi)$ is higher than the horizontal line through the tail with largest $y$-coordinate among the edges in $E(\chi')$. Between an UPSE and the next one, at most two paths are traversed (one to backtrack and one to again reach $n_{\mathcal T}$), hence the number of edges of $\mathcal D$ that are traversed between an UPSE and the next one is $\bigO(n)$. The total number of edges of $G$ which are deleted from or added to the current embedding when traversing such paths is in $\bigO(n)$, given that the size of $G$ is $\bigO(n)$. Hence, the delay of our algorithm is $\bigO(n)$. We get the following.

\begin{theorem}\label{th:st-enumeration}
Let $G$ be a $n$-vertex planar $st$-graph whose maximum $st$-cut has size $k$ and let $S$ be a set of $n$ points. It is possible to enumerate all \UPSEs of $G$ on $S$ with $\bigO(n)$ delay, using $\bigO(k n^{4k}\log n)$ space, after $\bigO(k n^{4k}\log n)$ set-up time. 
\end{theorem}

\section{Planar st-Graphs Composed of Two st-Paths}\label{sec:2-paths}

In this section, we discuss a special, and in our opinion interesting, case of \cref{th:k-paths}, namely the one in which the underlying graph of the given planar $st$-graph is an $n$-vertex cycle. Applying \cref{th:k-paths} to this setting would yield an $\bigO(n^8)$-time UPSE testing algorithm. Now, based on a characterization of the positive instances, we give a much faster algorithm for this case, provided that the points of $S$ are in general position.

\begin{theorem} \label{th:two-paths}
Let $G$ be an $n$-vertex planar $st$-graph consisting of two $st$-paths $P_L$ and $P_R$, and let $S$ be a pointset with $n$ points in general position. We have that $G$ admits an \UPSE on $S$ with $P_L$ to the left of $P_R$ if and only if
$|P_L| \geq |\convL(S)|$ and $|P_R| \geq |\convR(S)|$.
Also, it can be tested in $\bigO(n\log n)$ time whether $G$ admits an \UPSE on $S$. 

\end{theorem}
\begin{proof}

Provided the characterization in the statement holds, we can easily test whether $G$ admits an \UPSE on $S$ as follows. First, we compute the convex hull $\conv(S)$ of $S$, which can be done in $\bigO(n\log n)$ time. Second, we derive the sets $\convL(S)$ and $\convR(S)$, which can be done in $O(n)$ time by scanning $\conv(S)$. Finally, we compare the sizes of $\convL(S)$ and $\convR(S)$ with the ones of $P_L$ and $P_R$, which can be done in $O(1)$ time.  Therefore, in the following we focus on proving the characterization.

For the necessity, suppose for a contradiction that there exists an \UPSE on $S$ with $P_L$ to the left of $P_R$ and that $|P_L| < |\convL(S)|$; the case in which $|P_R| < |\convR(S)|$ is analogous. Since $| P_L \cup P_R| = n$, a vertex $d$ of~$P_R$ must be drawn on a point in $\convL(S)$. Consider the subpath $P_d$ of $P_R$ between $s$ and $d$. The drawing of~$P_d$ splits $\conv(S)$ into two closed regions, to the left and to the right of $P_d$. In any UPSE of $G$ on $S$ with~$P_L$ to the left of $P_R$, we have that $P_L$ lies in both regions, namely it lies in the region to the left of $P_d$ with the edge incident to $s$ and it lies in the region to the right of $P_d$ at $t$. Hence, the drawing of $P_L$ crosses the drawing of $P_d$, and thus the one of $P_R$, a contradiction.

In the following, we prove the sufficiency by induction on the size of $S$ (and, thus, of~$V(G)$). We give some preliminary definitions; see \cref{fig:2paths-base-case,fig:2paths-case-A1,fig:2paths-case-B2}.
Let $p_s$ and $p_t$ be the south and north extreme of $S$, respectively.
Consider the line $\ell_{st}$ through $p_s$ and $p_t$. 
Let $S_L$ ($S_R$) be the set consisting of the points of $S$ lying in the closed half-plane delimited by $\ell_{st}$ that includes all points that lie to the left (resp.\ right) of $\ell_{st}$, including $p_s$ and $p_t$.
Note that $\convL(S) \subseteq S_L$ and $\convR(S) \subseteq S_R$. Moreover, since $S$ is in general position, it holds that $S_L \cap S_R = \{p_s,p_t\}$.

\begin{figure}[t]
    \centering
        \includegraphics[scale=.6, page=1]{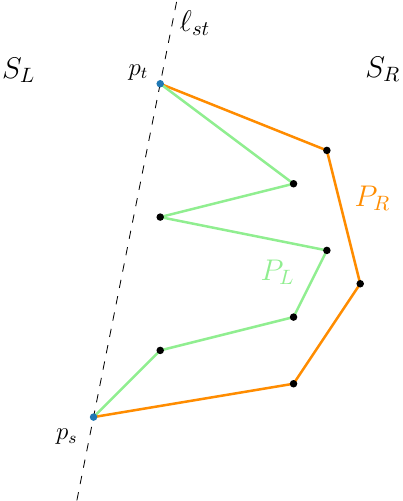}
    \caption{Illustration for the base case of \cref{th:two-paths}, when $S_L = \{p_s, p_t\}$ and $|\convR(S)| = |P_R|$. The drawing of $P_R$ coincides with $\mathcal E_R(S)$.}
        \label{fig:2paths-base-case}
\end{figure}

In the base case, it either holds that (a) $S_L = \{p_s, p_t\}$ and $|\convR(S)| = |P_R|$, or (b) $S_R= \{p_s, p_t\}$ and $|\convL(S)| = |P_L|$. We discuss the former case (see~\cref{fig:2paths-base-case}), as the latter case is symmetric.
In this case, an \UPSE~$\Gamma$ of $G$ on $S$ clearly exists and is, in fact, unique. In particular, the drawing of $P_R$ in~$\Gamma$ coincides with the right envelope $\mathcal{E}_R(S)$ of $S$, while the drawing of $P_L$ in $\Gamma$ is the $y$-monotone polyline that assigns to the $j$-th internal vertex of $P_L$ (when traversing $P_L$ from $s$ to $t$) the point of $S_R \setminus \convR(S)$ with the $j$-th smallest $y$-coordinate. Since each of such paths is $y$-monotone, it is not self-crossing. Also, no edge of $P_L$ crosses an edge of $P_R$, as the drawing of $P_R$ in $\Gamma$ coincides with $\mathcal{E}_R(S)$. 

If the base case does not hold, then we distinguish two cases based on whether both $S_L$ and $S_R$ contain a vertex different from $p_s$ and $p_t$ (\textbf{Case A}), or only one of them does (\textbf{Case~B}). In the following, we assume that in \textbf{Case B} the set $S_R$ contains a vertex different from $p_s$ and $p_t$, the case in which only $S_L$ contains a vertex different from $p_s$ and $p_t$ can be treated symmetrically. More formally, in \textbf{Case A} we have that $\{p_s,p_t\} \subset S_L$ and $\{p_s,p_t\} \subset S_R$, whereas in \textbf{Case~B} we have that $S_L = \{p_s,p_t\}$ and $\{p_s,p_t\} \subset S_R$. Note that, in \textbf{Case~B}, since the conditions of the base case do not apply and by the hypothesis of the statement, we have that $|P_R| > |\convR(S)|$ holds.

\begin{figure}[b!]
    \centering
    \begin{subfigure}[b]{0.3\textwidth}
    \centering
        \includegraphics[width=\textwidth, page=2]{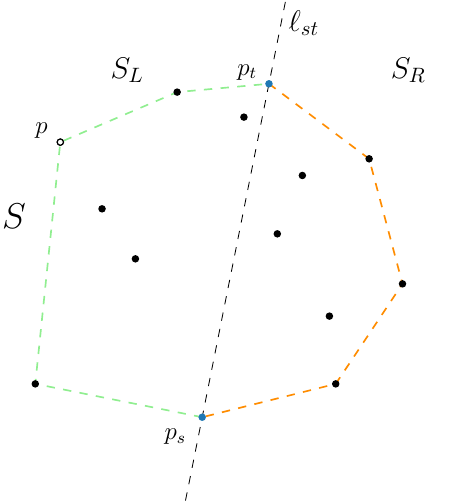}
        \subcaption{}
        \label{fig:2paths-case-A1-setup}
    \end{subfigure}
    \hfill
    \begin{subfigure}[b]{0.3\textwidth}
    \centering
        \includegraphics[width=\textwidth, page=3]{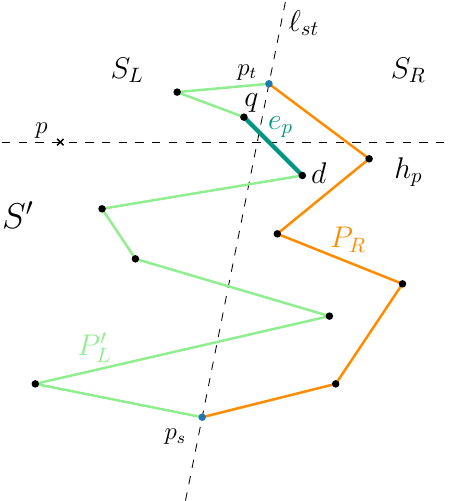}
        \subcaption{}
        \label{fig:2paths-case-A1-before}
    \end{subfigure}
    \hfill
    \begin{subfigure}[b]{0.3\textwidth}
    \centering
        \includegraphics[width=\textwidth, page=4]{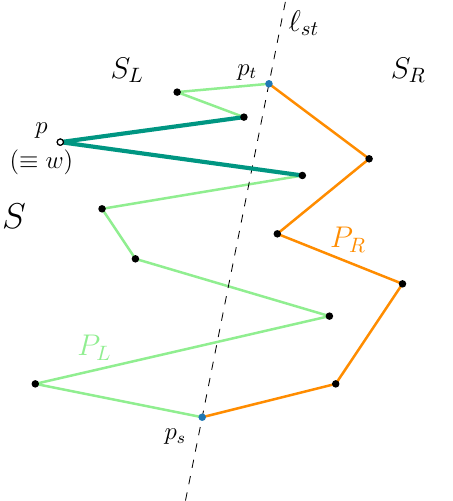}
        \subcaption{}
        \label{fig:2paths-case-A1-after}
    \end{subfigure}
    \caption{Illustrations for {\bf Case A1} in the proof of \cref{th:two-paths}.
    (a) $\convL(S)$ contains a point $p$ different from $p_s$ and $p_t$. (b) An UPSE $\Gamma'$ of the graph $G' = P'_L \cup P_R$  on the pointset $S' = S \setminus \{p\}$. (c) The UPSE $\Gamma$ of $G$ on $S$ obtained by the $p$-leftward-outer-extension of $\Gamma'$.}
    \label{fig:2paths-case-A1}
\end{figure}

If \textbf{Case A} holds, we distinguish two subcases. In \textbf{Case A1}, it holds $|P_L| \geq |S_L|$, whereas
in \textbf{Case A2}, it holds $|P_L| < |S_L|$. We discuss \textbf{Case A1} (see \cref{fig:2paths-case-A1}); \textbf{Case A2} can be treated symmetrically, given that in this case it holds that $|P_R| \geq |S_R|$. 

Suppose that {\bf Case A1} holds true. Then $\convL(S)$ contains 
a point $p$ different from $p_s$ and $p_t$; see \cref{fig:2paths-case-A1-setup}. Since by the hypotheses of this case $|P_L| \geq |S_L| \geq |\convL(S)|$ and $|\convL(S)| \geq 3$, we have that $P_L$ contains at least one internal vertex. Let $S' = S \setminus \{p\}$, let $P'_L$ be an $st$-path with $|P'_L|=|P_L|-1$, and let  $G'$ be the planar $st$-graph $P'_L \cup P_R$.
Since $|\convL(S')| \leq |S_L| - 1$  and since $|S_L| \leq |P_L|$, we have that $|\convL(S')| \leq |P_L|-1 = |P'_L|$. 
Thus, the graph $G'$ and the pointset $S'$ satisfy the conditions of the statement. 
By induction, we have that $G'$ admits an \UPSE $\Gamma'$ on~$S'$, see \cref{fig:2paths-case-A1-before}. 

We show how to modify $\Gamma'$ to obtain an \UPSE $\Gamma$ of $G$ on $S$ as follows; see \cref{fig:2paths-case-A1-before,fig:2paths-case-A1-after}. The drawing of $P_R$ is the same in $\Gamma$ as in $\Gamma'$. 
Let $h_p$ be the horizontal line passing through $p$. Since $\Gamma'$ is an \UPSE of $G'$ on $S'$ and since $y(p_s) < y(p) < y(p_t)$, we have that $h_p$ intersects the drawing of $P'_L$ in a single point. Such a point belongs to a segment that is the image of an edge $e_p$ of $P'_L$. Let $d$ and $q$ be the extremes of such a segment that are the images of the tail and of the head of $e_p$ in $\Gamma'$, respectively. We show how to modify the drawing of $P'_L$ to obtain a $y$-monotone drawing of $P_L$ that does not intersect $P_R$.
To this aim, we replace the drawing of $e_p$ with the $y$-monotone polyline composed of the segments $\overline{dp}$ and $\overline{pq}$. Note that such a polyline lies in the interior of the region delimited by the segment $\overline{d q}$ (representing $e_p$) and by the horizontal rays originating at $d$ and $q$ and directed leftward. Due to the fact that $P'_L$ is represented as a $y$-monotone polyline in $\Gamma'$, such a region is not traversed by the drawing of any edge. Thus, $\Gamma$ is an \UPSE of $G$ on $S$. 
We refer to the described procedure as the \emph{$p$-leftward-outer-extension} of $\Gamma'$; a \emph{$p$-rightward-outer-extension} of $\Gamma'$ is defined symmetrically.

If \textbf{Case B} holds, recall that $S_L = \{p_s,p_t\} \subset S_R$, and since the base case does not apply, we have that $|P_R| > |\convR(S)|$. Let $p$ be any point in $\convR(S) \setminus \{p_s,p_t\}$ and $S' = S \setminus \{p\}$. 
By the conditions of \textbf{Case B}, the path $P_R$ contains at least one internal vertex. We let $P'_R$ be an $st$-path with $|P'_R|=|P_R|-1$, and we let $G'$ be the $st$-graph $P_L \cup P'_R$.
We distinguish two cases based on the size of $\convR(S')$.
In \textbf{Case B1}, it holds $|P'_R| \geq |\convR(S')|$, whereas in \textbf{Case B2}, it holds $|P'_R| < |\convR(S')|$.

\begin{figure}[t!]
    \centering
    \begin{subfigure}[b]{0.3\textwidth}
    \centering
        \includegraphics[width=\textwidth, page=5]{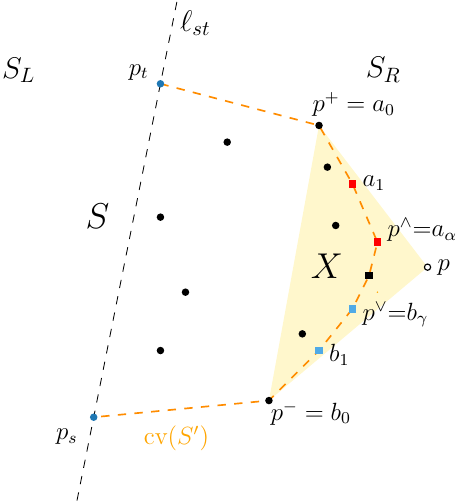}
        \subcaption{}
        \label{fig:2paths-case-B2-setup}
    \end{subfigure}
    \hfill
    \begin{subfigure}[b]{0.3\textwidth}
    \centering
        \includegraphics[width=\textwidth, page=7]{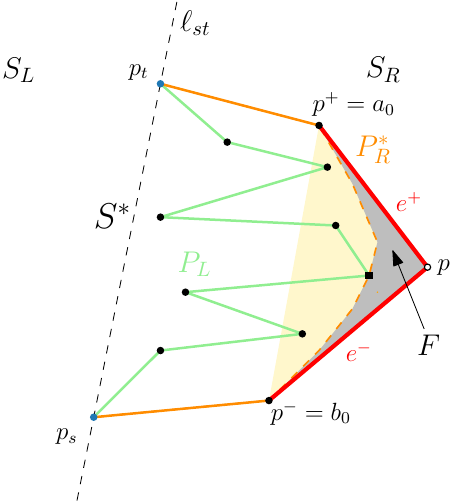}
        \subcaption{}
        \label{fig:2paths-case-B2-before}
    \end{subfigure}
    \hfill
    \begin{subfigure}[b]{0.3\textwidth}
    \centering
        \includegraphics[width=\textwidth, page=6]{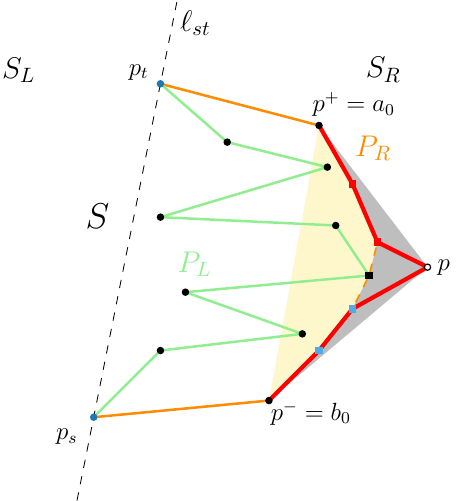}
        \subcaption{}
        \label{fig:2paths-case-B2-after}
    \end{subfigure}
    \caption{Illustrations for {\bf Case B2} in the proof of \cref{th:two-paths}. In this case, $S_L = \{p_s,p_t\}$, $|P_R| > |\convR(S)|$, and $|P'_R| < |\convR(S')|$ hold. (a) The triangle $\Delta p^+ p p^-$ is shaded yellow. (b) An UPSE $\Gamma^*$ of the graph $G^* = P_L \cup P^*_R$ on the pointset $S^* = S \setminus X^*$. (c) The UPSE of $G$ on $S$ obtained from $\Gamma^*$ by modifying the drawing of $P^*_R$ between $p^-$ and $p^+$ to use the points in $X^*$.
    }
    \label{fig:2paths-case-B2}
\end{figure}

In \textbf{Case B1}, the pair $(G',S')$ satisfies the conditions of the statement. In particular, it either matches the conditions of the base case or again those of {\bf Case B}.
Thus, since $|S'| = |S| - 1$ (and $|V(G')| = |V(G)| - 1$), we can inductively construct an \UPSE $\Gamma'$ of $G'$ on $S'$, and obtain an \UPSE of $G$ on $S$ via a $p$-rightward-outer-extension of $\Gamma'$.

In \textbf{Case B2}, which is the most interesting, we proceed as follows; see \cref{fig:2paths-case-B2}. Let $p^+$ be the point of $\convR(S)$ with the smallest $y$-coordinate and above $p$ and let $p^-$ be the point of $\convR(S)$ with the largest $y$-coordinate and below $p$. Let $X$ be the set of points of $S$ that lie in the interior of the triangle $\Delta p^+ p p^-$, including $p^+$ and $p^-$ and excluding $p$. Clearly, the right envelope of $\conv(X)$ forms a subpath of the right envelope of $\conv(S')$.
The set $\convR(X)$ consists of $p^-$, $p^+$, and of $k$ vertices not belonging to $\convR(S)$, depicted as squares in  \cref{fig:2paths-case-B2-setup}. Denote by $k^* = |P_R| - |\convR(S)|$ the number of points in the interior of $\conv(S)$ that need to be the image of a vertex of $P_R$ in an \UPSE of $G$ on $S$. Observe that $k > k^* > 0$ holds true. Indeed, $k^* > 0$ holds true since $(G,S)$ does not satisfy the conditions of the base case, and $k > k^*$ holds true since $(G,S)$ does not satisfy the conditions of {\bf Case B1}.
Let $p^{\wedge}$ be the point of $\convR(S')$ with the smallest $y$-coordinate and above $p$, and let $p^{\vee}$ be the point with the largest $y$-coordinate and below $p$.
Up to renaming, let $a_0 = p^+, a_1,\dots, a_\alpha = p^{\wedge}$ be  the subsequence of points of $\mathcal{E}_R(X)$ encountered when traversing $\mathcal{E}_R(X)$ from $p^+$ to $p^{\wedge}$ and observe that these points have decreasing $y$-coordinates. 
Similarly, let $b_0 = p^-, b_1,\dots, b_\gamma = p^{\vee}$ be the subsequence of points of $\mathcal{E}_R(X)$ encountered when traversing $\mathcal{E}_R(X)$ from $p^-$ to $p^{\vee}$  and observe that these points have increasing $y$-coordinates. 
We let the set $X^* \subset \convR(X)$ be $X^* =X^*_{\wedge} \cup X^*_{\vee}$, where $X^*_{\wedge}$ and $X^*_{\vee}$ are defined, based on the value of $k^*$, as follows. If $k^*\leq \alpha$, then let $X^*_{\wedge}= \{a_i| 1 \leq i \leq k^*\}$ and $X^*_{\vee}=\varnothing$, otherwise let $X^*_{\wedge}= \{a_i| 1 \leq i \leq \alpha\}$ and $X^*_{\vee} = \{b_i| 1 \leq i \leq k^*-\alpha\}$. 


%

\noindent
Observe that $|X^*| = k^*$. Also, by the definition of $k^*$, the path $P_R$ contains $\convR(S) -2 + k^*$ internal vertices and since $\convR(S)\geq 3$ in {\bf Case B}, we have that $P_R$ contains at least $k^* + 1$ internal vertices. 


Let $S^* = S \setminus X^*$, let $P^*_R$ be an $st$-path with $|P_R|-k^*$ vertices, and let $G^*$ be the $st$-graph $P_L \cup P^*_R$. We have that the pair $(G^*,S^*)$ satisfies the conditions of the statement, and in particular the base case. In fact,  $|P^*_R| = |P_R| - k^*$, and by the definition of $k^*$, we have that $|P_R| - k^* = |\convR(S)|$. Moreover,  by construction, $\convR(S)  = \convR(S^*)$, since the vertices of $X^*$ lie in the interior of $\conv(S)$. Thus, since $|S^*| = |S| - k^*$, by induction $G^*$ admits an \UPSE $\Gamma^*$ on $S^*$; see \cref{fig:2paths-case-B2-before}. 

We now show how to transform $\Gamma^*$ into an \UPSE $\Gamma$ of $G$ on $S$. Since the base case applies to $(G^*,S^*)$, we have that the endpoints of the edges of $P^*_R$ are consecutive along~$\mathcal{E}_R(S)$. In particular, there exist two adjacent edges $e^-$ and $e^+$ of $P^*_R$ such that the tail of $e^-$ is mapped to $p^-$, the head of $e^-$, which is the tail of $e^+$, is mapped to $p$, and the head of $e^+$ is mapped to $p^+$. Therefore, the \UPSE $\Gamma$ of $G$ on $S$ can be obtained from $\Gamma^*$ as follows; see \cref{fig:2paths-case-B2-after}.  We initialize $\Gamma= \Gamma^*$. 
The drawing of $P_L$ is the same in $\Gamma$ as in $\Gamma^*$. Next, we show how to modify the drawing of $P^*_R$ to obtain a $y$-monotone drawing $P_R$ that does not intersect the drawing of $P_L$ and uses the same points as $P^*_R$ and the points in $X^*$.
To this aim, we replace the drawing of $e^+$ with the (unique) $y$-monotone polyline connecting $p$ and $p^+$ that passes through all the points in $X^*_{\wedge}$. 
Also, we replace the drawing of $e^-$ with the (unique) $y$-monotone polyline connecting $p^-$ and $p$ that passes through all the points in $X^*_{\vee}$; note that $X^*_{\vee}$ might be empty, in which case the polyline still coincides with the drawing of $e^-$. This concludes the construction of $\Gamma$. To see that $\Gamma$ is an \UPSE of $G$ on $S$ observe that the above polylines (i) are each non-self-crossing, as they are $y$-monotone, (ii) do not cross with each other as they entirely lie either above or below $p$ (and only meet at $p$), and (iii) do not cross any edge of $\Gamma'$ as they lie in the region $F$ (shaded gray in \cref{fig:2paths-case-B2-before,fig:2paths-case-B2-after}) obtained by subtracting from the triangle $\Delta p^+ p  p^-$ (interpreted as a closed region) all the points of $\conv(X)$. Indeed, observe that in $\Gamma^*$, the region $F$ is not traversed by any edge and that the only points of $S^*$ that lie on the boundary of $F$ are $p$ and the points in $\convR(X) \setminus X^*$. 
\end{proof}

\section{Enumerating Non-Crossing Monotone Hamiltonian Cycles}\label{se:geometry}

\cref{th:two-paths} allows us to test whether an $n$-vertex planar $st$-graph $G$ composed of two $st$-paths can be embedded as a non-crossing monotone Hamiltonian cycle on a set $S$ of $n$ points. We now show an efficient algorithm for enumerating {\em all} the non-crossing monotone Hamiltonian cycles on $S$. \cref{fig:exponential-monotone-cycles} shows two non-crossing monotone Hamiltonian cycles on a pointset.

\begin{figure}[t!]
    \centering
    \begin{subfigure}[b]{0.3\textwidth}
    \centering
        \includegraphics[scale=.75, page=10]{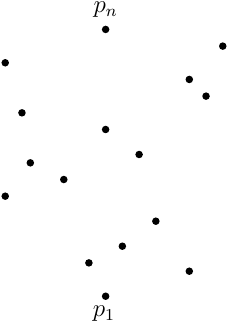}
        \subcaption{}
    \end{subfigure}
    \hfill
    \begin{subfigure}[b]{0.3\textwidth}
    \centering
        \includegraphics[scale=.75, page=11]{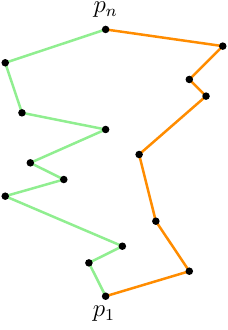}
        \subcaption{}
    \end{subfigure}
    \hfill
    \begin{subfigure}[b]{0.3\textwidth}
    \centering
        \includegraphics[scale=.75, page=12]{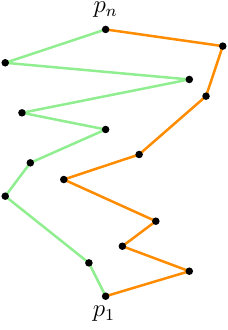}
        \subcaption{}
    \end{subfigure}
    \caption{Two non-crossing monotone Hamiltonian cycles on the same pointset.}
    \label{fig:exponential-monotone-cycles}
\end{figure}

\begin{theorem}\label{th:hamiltonian-enumeration}
Let $S$ be a set of $n$ points. It is possible to enumerate all the non-crossing monotone Hamiltonian cycles on $S$ with $\bigO(n)$ delay, using $\bigO(n^2)$ space, after $\bigO(n^2)$ set-up time.  
\end{theorem}
Let $p_1,\dots,p_{n}$ be the points of $S$, ordered by increasing $y$-coordinates. This order can be computed in $\bigO(n\log n)$ time. For $i\in[n]$, let $S_i=\{p_1,\dots,p_i\}$. A \emph{bipath} $B$ on $S_i$ consists of two non-crossing monotone paths $L$ and $R$ on $S_i$, each of which might be a single point, such that (see \cref{fig:bipaths}):
\begin{enumerate}
    \renewcommand{\theenumi}{\roman{enumi}} 
    \renewcommand{\labelenumi}{\textbf{(\theenumi)}} 
\item $L$ and $R$ start at $p_1$;
\item each point of $S_i$ is the image of an endpoint of a segment of $B$; and
\item if $L$ and $R$ both have at least one segment, then $L$ is to the left of $R$.
\end{enumerate}

\begin{figure}[t!]
    \centering
    \begin{subfigure}[b]{0.3\textwidth}
    \centering
        \includegraphics[scale=.75, page=1]{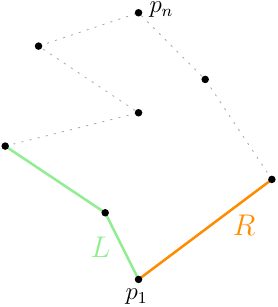}
        \subcaption{}
    \end{subfigure}
    \hfill
    \begin{subfigure}[b]{0.3\textwidth}
    \centering
        \includegraphics[scale=.75, page=2]{img/HamiltonianCycle.pdf}
        \subcaption{}
    \end{subfigure}
    \hfill
    \begin{subfigure}[b]{0.3\textwidth}
    \centering
        \includegraphics[scale=.75, page=3]{img/HamiltonianCycle.pdf}
        \subcaption{}
    \end{subfigure}
    \caption{Three bipaths on $S_4$. The first two bipaths are extensible, while the third one is not.
    Dotted lines complete non-crossing Hamiltonian cycles on $S$ whose restriction to $S_4$ is the bipath.}
    \label{fig:bipaths}
\end{figure}

We say that a bipath $B$ is \emph{extensible} if there exists a non-crossing monotone Hamiltonian cycle on $S$ whose restriction to $S_i$ is $B$. 
Consider a bipath $B$ on $S_i$ with $1<i<n$. Let $p_{\ell(B)}$ and $p_{r(B)}$ be the endpoints of $L$ and $R$ with the highest $y$-coordinate, respectively. First, suppose that $\ell(B)> r(B)$, that is, $p_{\ell(B)}$ is higher than $p_{r(B)}$. Then note that $\ell(B)=i$; also, it might be that $r(B)=1$. Consider the ray $\rho(p_{r(B)},S_{\ell(B)}\setminus S_{r(B)})$; recall that this is the rightmost ray starting at $p_{r(B)}$ and passing through a point of $S_{\ell(B)}\setminus S_{r(B)}$. We denote by $\mathcal{R}(B)$ the open region of the plane strictly to the right of $\rho(p_{r(B)},S_{\ell(B)}\setminus S_{r(B)})$ and strictly above the horizontal line through $p_{\ell(B)}$; see \cref{fig:rb}. Similarly, if $p_{r(B)}$ is higher than $p_{\ell(B)}$, then $\mathcal{L}(B)$ is the open region of the plane strictly to the left of the leftmost ray $\ell(p_{\ell(B)},S_{r(B)}\setminus S_{\ell(B)})$ from $p_{\ell(B)}$ through a point of $S_{r(B)}\setminus S_{\ell(B)}$ and strictly above the horizontal line through $p_{r(B)}$; see \cref{fig:lb}.

\begin{figure}[t!]
    \centering
    \begin{subfigure}[b]{0.45\textwidth}
    \centering
        \includegraphics[scale=.75, page=4]{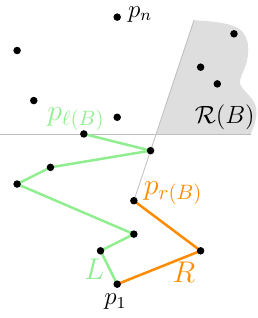}
        \subcaption{}
    \label{fig:rb}
    \end{subfigure}
    \hfill
    \begin{subfigure}[b]{0.45\textwidth}
    \centering
        \includegraphics[scale=.75, page=5]{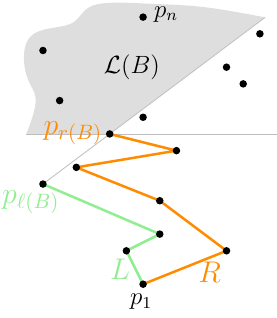}
        \subcaption{}
    \label{fig:lb}
    \end{subfigure}
    \caption{(a) Region $\mathcal{R}(B)$ for a bipath $B$ with $\ell(B)>r(B)$. (b) Region $\mathcal{L}(B)$ for a bipath $B$ with $r(B)>\ell(B)$.}
    \label{fig:regions}
\end{figure}

For any $i\in [n-1]$, we say that a bipath $B$ on $S_i$ is \emph{safe} if:
\begin{enumerate}
    \renewcommand{\theenumi}{\roman{enumi}} 
    \renewcommand{\labelenumi}{\textbf{(\theenumi)}} 
    \item $i=1$; or
    \item $i>1$, $p_{\ell(B)}$ is higher than $p_{r(B)}$, and $|\mathcal{R}(B)\cap S|\geq 1$; or 
    \item $i>1$, $p_{r(B)}$ is higher than $p_{\ell(B)}$, and $|\mathcal{L}(B)\cap S|\geq 1$.
\end{enumerate}
We have the following lemma.
\begin{lemmax} \label[lemmax]{le:safe-is-extensible}
A bipath $B$ is extensible if and only it is safe. 
\end{lemmax}

\begin{proof}
    First, we prove the necessity. Suppose that $B$ is extensible and let $C$ be any non-crossing monotone Hamiltonian cycle on $S$ whose restriction to $S_i$ is $B$. Also suppose, for a contradiction, that $B$ is not safe, which implies that $i>1$. Assume that $p_{\ell(B)}$ is higher than $p_{r(B)}$, as the other case is symmetric. Then we have $\mathcal{R}(B) \cap S = \varnothing$. Let $\overline{p_{r(B)}p'_{r(B)}}$ be the segment of $C$ such that $y(p'_{r(B)}) > y(p_{r(B)})$. Since all points in $S_{\ell(B)}\setminus S_{r(B)}$ belong to $L$, we have $p'_{r(B)}$ lies strictly above the horizontal line through $p_{\ell(B)}$. This, together with the fact that 
    $S$ contains no point strictly above the horizontal line through $p_{\ell(B)}$ and to the right of the ray $\rho(p_{r(B)},S_{\ell(B)}\setminus S_{r(B)})$, implies that the ray $\rho(p_{r(B)}p'_{r(B)})$ lies to the left of the ray $\rho(p_{r(B)},S_{\ell(B)}\setminus S_{r(B)})$, which implies that the segment $(p_{r(B)},p'_{r(B)})$ crosses the path $L$, a contradiction to the fact that $C$ is non-crossing.   

\begin{figure}[t!]
    \centering
    \begin{subfigure}[b]{0.45\textwidth}
    \centering
        \includegraphics[scale=.75, page=6]{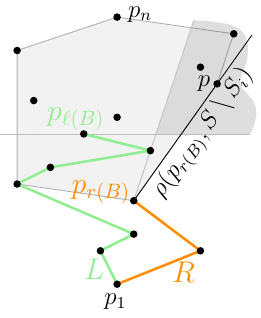}
    \end{subfigure}
    \hfill
    \begin{subfigure}[b]{0.45\textwidth}
    \centering
        \includegraphics[scale=.75, page=7]{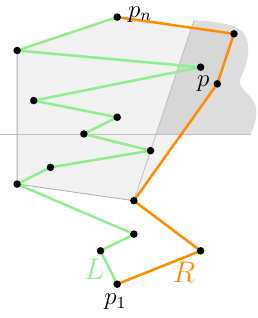}
    \end{subfigure}
    \caption{Since the point $p$ on the ray $\rho(p_{r(B)},S\setminus S_i)$ defines a segment $\overline{p_{r(B)}p}$ which is on the boundary of the convex hull of $S \setminus S_{r(B)-1}$ (the convex hull is shaded light-gray), we can complete $R$ via the boundary of the convex hull and $L$ via the remaining points.}
    \label{fig:safe-extensible}
\end{figure}
    Second, we prove the sufficiency. Suppose that $B$ is safe. We show how to construct a non-crossing monotone Hamiltonian cycle $C$ on $S$ whose restriction to $S_i$ is $B$.  
    Assume first that $i=1$. Then $C$ can be constructed as the union of two monotone paths. The first path is one of the two paths between $s$ and $t$ on the boundary of the convex hull of $S$. The second path from $s$ to $t$ traverses all the points of $S$ that are not on the first path, in increasing order of $y$-coordinate. 
    Assume next that $i>1$ and refer to \cref{fig:safe-extensible}. Assume also that $p_{\ell(B)}$ is higher than $p_{r(B)}$, as the other case is symmetric. Then $\mathcal{R}(B)$ contains some points of $S$. Consider the rightmost ray $\rho(p_{r(B)},S\setminus S_i)$ starting from $p_{r(B)}$ and passing through a point $p$ in $S\setminus S_i$. Observe that $\overline{p_{r(B)}p}$ is a segment on the boundary of the convex hull of $S \setminus S_{r(B)-1}$. Hence, we can augment $R$ so that it becomes a monotone path from $s$ to $t$, by adding to it the part of the boundary of the convex hull of $S \setminus S_{r(B)-1}$ from $p_{r(B)}$ to $t$ (by proceeding in counter-clockwise direction on this boundary from $p_{r(B)}$ to $t$). Also, we can augment $L$ so that it becomes a monotone path from $s$ to $t$ by making it pass through all the points in $S\setminus S_i$ that are not used by $R$, and finishing at $t$.
\end{proof}

We now describe our algorithm. The algorithm implicitly defines and explores a search tree $T$. The leaves of $T$ have level $n$ and correspond to non-crossing monotone Hamiltonian cycles on $S$. The internal nodes at level $i$ correspond to extensible bipaths on $S_i$ and have at most two children each. The exploration of $T$ performed by our enumeration algorithm is a depth-first traversal. When a node $\mu$ is visited, the number of its children is established. If $\mu$ has at least one child, the visit proceeds with any child of $\mu$. Otherwise, $\mu$ is a leaf; then the visit proceeds with any unvisited child of the ancestor of $\mu$ that has largest level, among the ancestors of $\mu$ that have unvisited children.

\begin{itemize}
    \item The algorithm starts at the root of $T$, which corresponds to the (unique) safe bipath on~$S_1$. 

    \item At each node $\mu$ at level $i\in [n-2]$ of $T$, corresponding to a bipath $B(\mu)$, we construct either one or two bipaths on $S_{i+1}$, associated with either one or two children of $\mu$, respectively. Let $L(\mu)$ and $R(\mu)$ be the left and right non-crossing monotone paths composing $B(\mu)$, respectively, and let $p_{\ell(B(\mu))}$ and $p_{r(B(\mu))}$ be the endpoints of $L(\mu)$ and $R(\mu)$ with the highest $y$-coordinate, respectively. If $\overline{p_{\ell(B(\mu))} p_{i+1}}$ does not cross $R(\mu)$, then let $B_L=B(\mu)\cup \overline{p_{\ell(B(\mu))} p_{i+1}}$. We test whether $B_L$ is a safe bipath and, in the positive case, add to $\mu$ a child $\mu_L$ corresponding to $B_L$. Analogously, if $\overline{p_{r(B(\mu))} p_{i+1}}$ does not cross $L(\mu)$, then we test whether $B_R=B(\mu)\cup \overline{p_{r(B(\mu))} p_{i+1}}$ is a safe bipath and, in the positive case, add to $\mu$ a child $\mu_R$ corresponding to $B_R$. Note that the algorithm guarantees that each node at a level smaller than or equal to $n-1$ of $T$ is safe, and thus, by \cref{le:safe-is-extensible}, extensible.

    \item Finally, at each node $\mu$ at level $n-1$, we add a leaf $\lambda$ to $\mu$ corresponding to the non-crossing monotone Hamiltonian cycle $B(\mu)\cup \overline{p_{\ell(B(\mu))} p_n} \cup \overline{p_{r(B(\mu))} p_n}$. Note that, since $\mu$ is extensible, such a cycle is indeed non-crossing.
\end{itemize}

In order to complete the proof of \cref{th:hamiltonian-enumeration}, we show what follows: 
\begin{enumerate}
    \renewcommand{\theenumi}{\roman{enumi}} 
    \renewcommand{\labelenumi}{\textbf{(\theenumi)}} 
    \item \label{prop:level-leaves} Each node of $T$ at level $i\neq n$ is internal.
     \item \label{prop:leaf-cycle} Each leaf corresponds to a non-crossing monotone Hamiltonian cycle on $S$.
    \item \label{prop:distinct-cycle} Distinct leaves correspond to different non-crossing monotone Hamiltonian cycles on $S$.
    \item \label{prop:all-cycles} For each non-crossing monotone Hamiltonian cycle on $S$, there exists a leaf of $T$ corresponding to it.
    \item \label{prop:space-time} Using $\bigO(n^2)$ pre-processing time and $\bigO(n^2)$ space, the algorithm enumerates each non-crossing monotone Hamiltonian cycle on $S$ with $\bigO(n)$ delay.
\end{enumerate}

\begin{itemize}
\item  To prove \cref{prop:level-leaves}, we show that the leaves of $T$ have all level $n$. Consider a node $\mu$ of $T$ with level $i<n-1$, we prove that it has a child in $T$. Recall that $B(\mu)$ is safe, otherwise it would not had been added to $T$, and thus, by \cref{le:safe-is-extensible}, it is extensible. Hence, there exists a non-crossing monotone Hamiltonian cycle $C$ on $S$ whose restriction to $S_i$ is $B(\mu)$. Also, the restriction of $C$ to $S_{i+1}$ is a bipath $B'(\mu)$ on $S_{i+1}$ which coincides with $B(\mu)$, except that it contains either the segment $\overline{p_{\ell(B(\mu))} p_{i+1}}$ or the segment $\overline{p_{r(B(\mu))} p_{i+1}}$. Since $B'(\mu)$ is the restriction of $C$ to $S_{i+1}$, it is extensible and thus, by \cref{le:safe-is-extensible}, it is safe. It follows that $\mu$ has a child corresponding to $B'(\mu)$, which is inserted in $T$ when adding either the segment $\overline{p_{\ell(B(\mu))} p_{i+1}}$ or the segment $\overline{p_{r(B(\mu))} p_{i+1}}$ to $B(\mu)$. The proof that a node with level $n-1$ is not a leaf is analogous.

\item To prove \cref{prop:leaf-cycle}, consider a leaf $\lambda$ and its parent $\mu$ in $T$. Note that $\mu$ is associated with a safe bipath $B(\mu)$ on $S_{n-1}$; by \cref{le:safe-is-extensible}, we have that $B(\mu)$ is extensible. Since $B(\mu)$ is extensible, the (unique) monotone Hamiltonian cycle on $S$ whose restriction to $S_{n-1}$ is $B(\mu)$ is non-crossing. This cycle corresponds to $\lambda$ and is added to $T$ when visiting~$\mu$.

\item To prove \cref{prop:distinct-cycle}, suppose for a contradiction that there exist two leaves $\lambda_1$ and $\lambda_2$ associated with two monotone Hamiltonian cycles $C_1$ and $C_2$, respectively, with $C_1=C_2$. Let $\mu$ be the lowest common ancestor of $\lambda_1$ and $\lambda_2$ in $T$. Let $j$ be the level of $\mu$. Denote by $\mu_i$ the child of $\mu$ leading to $\lambda_i$, with $i\in \{1,2\}$. By the construction of $T$, we have that exactly one of the bipaths $B(\mu_1)$ and $B(\mu_2)$ contains the segment $\overline{p_{\ell(B(\mu))}p_{j+1}}$, while the other one contains the segment $\overline{p_{r(B(\mu))}p_{j+1}}$. This contradicts the fact that $C_1 = C_2$.

\item To prove \cref{prop:all-cycles}, let $C$ be a non-crossing monotone Hamiltonian cycle on $S$. Consider the safe bipath $B$ on $S_{n-1}$ obtained by removing from $C$ the point $p_n$, together with its two incident segments. It suffices to show that $T$ contains a node $\mu$ such that $B = B(\mu)$. In fact, in this case, $\mu$ is an extensible node of level $n-1$ whose unique child in $T$ is the leaf corresponding to $C$.
To prove that $T$ contains such a node $\mu$, we prove by induction that, for every level $i=1,\dots,n-1$, the tree $T$ contains a node corresponding to the restriction $B_i$ of $B$ to $S_i$.
The base case trivially holds.
For the inductive case, suppose that $T$ contains a node $\nu$ whose associated bipath $B(\nu)$ is $B_{i-1}$. Then $B_{i}$ is obtained by adding either the segment $\overline{p_{\ell(B(\nu))}p_i}$ or the segment $\overline{p_{r(B(\nu))}p_i}$ to $B_{i-1}$. Since $B_{i}$ is extensible, by \cref{le:safe-is-extensible} it is safe, and hence $\nu$ has a child in $T$ corresponding to $B_i$. 

\item Finally, we prove \cref{prop:space-time}. To this aim, we compute in $\bigO(n^2)$ time two tables $C$ and $D$. The first one allows us to quickly test whether a bipath on $S_i$ can be extended to a bipath on $S_{i+1}$ (so that no crossing is introduced). The second table allows us to quickly test whether a bipath on $S_i$ is safe. 

We first describe the computation of the table $C$, which has $\bigO(n^2)$ size, can be computed in $\bigO(n^2)$ time, and allows us to answer in $\bigO(1)$ time the following questions: Given a bipath $B$ on $S_i$ composed of the monotone $st$-paths $L$ and $R$ respectively ending at points $p_\ell$ and $p_r$, is $B\cup \overline{p_\ell p_{i+1}}$ a bipath on $S_{i+1}$ and is $B\cup \overline{p_r p_{i+1}}$ a bipath on $S_{i+1}$? That is, the table allows us to test whether the segment $\overline{p_\ell p_{i+1}}$ crosses any edge of $R$ and whether the segment $\overline{p_r p_{i+1}}$ crosses any edge of $L$.  

We only discuss how $C$ allows us to decide whether the segment $\overline{p_\ell p_{i+1}}$ crosses any edge of $R$, as the arguments for deciding whether the segment $\overline{p_r p_{i+1}}$ crosses any edge of $L$ are analogous. If $i=\ell$, then obviously the segment $\overline{p_\ell p_{i+1}}$ does not cross any edge of $R$, as it lies completely above $R$. So in the following we assume that $i=r$, that is, the point $p_\ell$ is lower than $p_r$, which is the highest point of $S_i$. This implies that $R$ contains the polyline $(p_{{\ell}+1}, p_{{\ell}+2}, \dots, p_r)$, as in \cref{fig:planar-test-success}.  

A key point for our efficient test is that whether $B\cup \overline{p_\ell p_{r+1}}$ is a bipath only depends on the points $p_{\ell},p_{{\ell}+1},\dots,p_r,p_{r+1}$, and not on the points lower than $p_{\ell}$. In particular, let $p_x$ be the point of $S_i$ with $x<{\ell}$ such that the segment $\overline{p_xp_{{\ell}+1}}$ belongs to $R$. Then the actual placement of $p_x$ does not matter for whether $\overline{p_\ell p_{r+1}}$ crosses $\overline{p_xp_{{\ell}+1}}$ or not, see \cref{fig:planar-test-fail}. This is formalized in the following claim.

\begin{figure}[t!]
    \centering
    \begin{subfigure}[b]{0.45\textwidth}
    \centering
        \includegraphics[scale=.75, page=13]{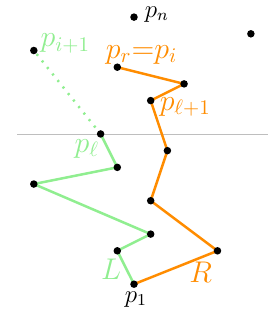}
        \subcaption{}
    \label{fig:planar-test-success}
    \end{subfigure}
    \hfill
    \begin{subfigure}[b]{0.45\textwidth}
    \centering
        \includegraphics[scale=.75, page=14]{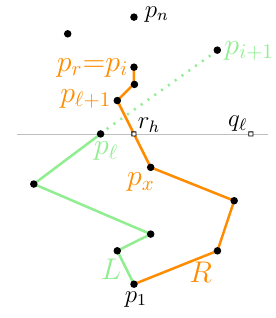}
        \subcaption{}
    \label{fig:planar-test-fail}
    \end{subfigure}
    \caption{Extensibility of a bipath $B$ whose monotone $st$-paths $L$ and $R$ end at points $p_\ell$ and $p_r$ with a segment $\overline{p_\ell p_{i+1}}$. In (a) the segment $\overline{p_\ell p_{i+1}}$ does not cross $B$, while in (b) it does.}
    \label{fig:crossing-test}
\end{figure}

\begin{claimx} \label[claimx]{cl:cross-past}
Let $B$ be a bipath on $S_i$ composed of two monotone $st$-paths $L$ and $R$ ending at points $p_\ell$ and $p_r$, where $\ell<r$, and  let $p_x$ be the point of $S_i$ with $x<{\ell}$ such that the segment $\overline{p_xp_{{\ell}+1}}$ belongs to $R$. Also, let $q_{\ell}$ be any point on the horizontal line $h_{\ell}$ through $p_{\ell}$, to the right of every point in $S$. Then the segment $\overline{p_\ell p_{r+1}}$ crosses $\overline{p_xp_{{\ell}+1}}$ if and only if it crosses $\overline{q_\ell p_{{\ell}+1}}$.        
\end{claimx}
\begin{proof}
Suppose that $\overline{p_\ell p_{r+1}}$ crosses $\overline{p_xp_{{\ell}+1}}$. We prove that $\overline{p_\ell p_{r+1}}$ crosses $\overline{q_\ell p_{{\ell}+1}}$, as well. The proof for the opposite direction is analogous. Let $r_h$ be the intersection point of $\overline{p_xp_{{\ell}+1}}$ with $h_{\ell}$. Since $\overline{p_xp_{{\ell}+1}}$ belongs to $R$, we have that $r_h$ lies to the right of $p_\ell$. This implies that, by rotating a ray $\rho(p_\ell,p_{{\ell}+1})$ starting from $p_\ell$ and passing through $p_{{\ell}+1}$ in clockwise direction, around $p_\ell$, the point $p_{r+1}$ is encountered before $r_h$. It follows that, by rotating $\rho(p_\ell,p_{{\ell}+1})$ in clockwise direction around $p_\ell$, the point $p_{r+1}$ is encountered before $q_{\ell}$, as well, since the ray starting at $p_\ell$ and passing through $q_{\ell}$ is the same as the ray starting at $p_\ell$ and passing through $r_h$. Hence, $\overline{p_\ell p_{r+1}}$ crosses $\overline{q_\ell p_{{\ell}+1}}$.
\end{proof}

A corollary of \cref{cl:cross-past} is that the segment $\overline{p_\ell p_{r+1}}$ crosses a bipath $B$ on $S_i$ composed of two monotone $st$-paths ending at points $p_\ell$ and $p_r$, with $\ell<r$, if and only if it crosses any other bipath $B'$ on $S_i$ composed of two monotone $st$-paths ending at points $p_\ell$ and $p_r$. This is obvious if the crossing involves a segment $\overline{p_yp_{y+1}}$, for some $y\in\{\ell+1,\ell+2,\dots,r-1\}$, as such a segment belongs both to $B$ and to $B'$, whereas it comes from \cref{cl:cross-past} if the crossing involves a segment $\overline{p_xp_{\ell+1}}$ of $B$ or $B'$ with $x<\ell$.

We are now ready to describe the table $C$ and its computation in greater detail. The table $C$ is indexed by triples $\langle p_\ell, p_r, X \rangle$, where $p_\ell$ and $p_r$ are distinct points in $S$ and $X \in \{L,R\}$. Note that $C$ has $\bigO(n^2)$ entries. Let $i=\max\{\ell,r\}$. The entry $C[p_\ell, p_r,L]$ is $\texttt{True}$ if and only if the segment $\overline{p_\ell p_{i+1}}$ does not cross any bipath $B$ on $S_i$ composed of two monotone $st$-paths $L$ and $R$ ending at points $p_\ell$ and $p_r$, if such a bipath exists, otherwise the value of $C[p_\ell, p_r,L]$ is irrelevant. Likewise, the entry $C[p_\ell, p_r,R]$ is $\texttt{True}$ if and only if the segment $\overline{p_r p_{i+1}}$ does not cross any bipath $B$ on $S_i$ composed of two monotone $st$-paths $L$ and $R$ ending at points $p_\ell$ and $p_r$, if such a bipath exists, otherwise the value of $C[p_\ell, p_r,R]$ is irrelevant. 

We show how to compute the entries $C[p_\ell, p_r,L]$, the computation of the entries $C[p_\ell, p_r,R]$ is done analogously. As discussed before, if $\ell>r$, then $C[p_\ell, p_r,L]=\texttt{True}$; this condition can be verified in $\bigO(1)$ time, hence in $\bigO(n^2)$ time over all entries of $C$. Assume now that $\ell<r=i$. A simple way of computing $C[p_\ell, p_r,L]$ would consist of verifying whether $\overline{p_\ell p_{r+1}}$ intersects any of the segments $\overline{q_\ell p_{\ell+1}},\overline{p_{\ell+1}p_{\ell+2}},\dots,\overline{p_{r-1}p_{r}}$. However, this would take $\Omega(r-\ell)$ time per entry, which would sum up to $\Omega(n^3)$ over all entries of $C$. Instead, for each fixed $\ell \in [n-2]$, we compute all the entries $C[p_\ell,p_r, L]$ with $r= \ell+1, \ell+2, \dots,n-1$ in overall $\bigO(n)$ time, as described below. This sums up to $\bigO(n^2)$ time over all the entries $C[p_\ell,p_r, L]$ of $C$ with $\ell=1,2,\dots,n-2$ and $r=\ell+1,\ell+2,\dots,n-1$.

Initialize a value $\alpha$ to the angle that is defined by a counter-clockwise rotation around $p_\ell$ of a horizontal ray starting at $p_\ell$ and directed rightward, so that the rotation stops when the ray passes through $p_{{\ell}+1}$. We now look at the values $j=\ell+1,\ell+2,\dots,n-1$ one by one. When we look at a value $j$, we compute the angle $\alpha_j$ that is defined by a counter-clockwise rotation  around $p_\ell$ of a horizontal ray starting at $p_\ell$ and directed rightward, so that the rotation stops when the ray passes through $p_{j+1}$. Two cases can happen. 
\begin{figure}[t!]
    \centering
    \begin{subfigure}[b]{0.45\textwidth}
    \centering
        \includegraphics[scale=.75, page=15]{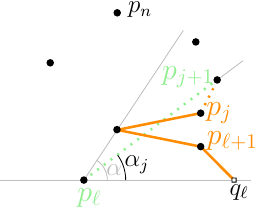}
        \subcaption{}
    \label{fig:rotation-fail}
    \end{subfigure}
    \hfill
    \begin{subfigure}[b]{0.45\textwidth}
    \centering
        \includegraphics[scale=.75, page=16]{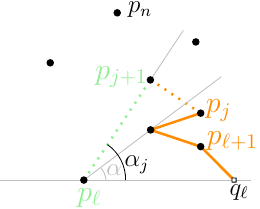}
        \subcaption{}
    \label{fig:rotation-success}
    \end{subfigure}
    \caption{Computation of the value of the entry $C[p_\ell, p_j,L]$. In (a) we have $\alpha_j\leq \alpha$, hence $C[p_\ell, p_j,L]=\texttt{False}$, while in (b) we have $\alpha_j> \alpha$, hence $C[p_\ell, p_j,L]=\texttt{True}$.}
    \label{fig:rotation}
\end{figure}
If $\alpha_j\leq \alpha$, as in \cref{fig:rotation-fail}, then we leave $\alpha$ unaltered and we set $C[p_\ell, p_j,L]=\texttt{False}$. Otherwise, that is, if $\alpha_j>\alpha$, as in \cref{fig:rotation-success}, then we set $\alpha$ to the value of  $\alpha_j$ and we set $C[p_\ell, p_j,L]=\texttt{True}$. 

Clearly, this computation takes $\bigO(1)$ per value~$j$, hence $\bigO(n)$ time for all the entries $C[p_\ell,p_r, L]$ with $r= \ell+1, \ell+2, \dots,n-1$, and thus  $\bigO(n^2)$ time over all the entries $C[p_\ell, p_r,L]$ of $C$. Concerning the correctness of the computed values, let $q_1,\dots,q_n$ be $n$ points such that, for $i=1,\dots,n$, the point $q_i$ has the same $y$-coordinate as $p_i$ and lies to the right of every point $p_j$ with $j=1,\dots,n$. It suffices to observe that the straight-line segment $\overline{p_\ell p_{j+1}}$ does not cross the polyline $(q_\ell,p_{\ell+1},p_{\ell+2},\dots,p_{j})$ if and only if a counter-clockwise rotation around $p_\ell$ of a horizontal ray starting at $p_\ell$ and directed rightward passes through all of $q_{\ell},p_{\ell+1},p_{\ell+2},\dots,p_{j}$ before passing through $p_{j+1}$. This is expressed by the condition $\alpha_j>\alpha$. As discussed before, assuming that a bipath $B$ on $S_j$ composed of two monotone $st$-paths ending at $p_{\ell}$ and $p_j$ exists, the straight-line segment $\overline{p_\ell p_{j+1}}$ crosses $B$ if and only if it crosses the polyline $(q_\ell,p_{\ell+1},p_{\ell+2},\dots,p_{j})$, from which the correctness of the computed entry values follows. 

We now turn our attention to the computation of the table $D$, which has $\bigO(n^2)$ size and allows us to test in $\bigO(1)$ time whether a bipath $B$ on $S_{i}$, with $i \in \{2,\dots,n-1\}$, is safe. 

The table $D$ is indexed by triples $\langle p_a, p_b, X \rangle$, where $p_a,p_b \in S$ with $a < b$ and $X \in \{L,R\}$. Each entry of $D$ contains a Boolean value $D[p_a,p_b, X]$ defined as follows. 

\begin{itemize}
    \item Suppose that $X = R$. Consider the rightmost ray $\rho(p_a,S_b \setminus S_a)$ starting from $p_a$ and passing through a point in $S_b \setminus S_a$. We denote by $\mathcal{R}(p_a,p_b)$ the open region of the plane strictly to the right of $\rho(p_a,S_b \setminus S_a)$ and strictly above the horizontal line  through $p_b$. 
Then, $D[p_a,p_b, R] = \texttt{True}$ if and only if $\mathcal{R}(p_a,p_b) \cap S \neq \varnothing$.
\item Next, suppose that $X = L$. Consider the leftmost ray $\ell(p_a,S_b \setminus S_a)$ starting from $p_a$ and passing through a point in $S_b \setminus S_a$. We denote by $\mathcal{L}(p_a,p_b)$ the open region of the plane strictly to the left of the ray $\ell(p_a,S_b \setminus S_a)$ and strictly above the horizontal line passing through $p_b$. 
Then, $D[p_a,p_b, L] = \texttt{True}$ if and only if $\mathcal{L}(p_a,p_b) \cap S \neq \varnothing$.
\end{itemize}

For each fixed $a \in [n-1]$, we show how to compute all the entries $D[p_a,p_b, R]$ with $b= a+1, a+2, \dots,n$ in overall $\bigO(n)$ time. This sums up to $\bigO(n^2)$ time over all the entries $D[p_a,p_b, R]$ of $D$ with $a=1,2,\dots,n-1$ and $b=a+1,a+2,\dots,n$. The computation of the entries $D[p_a,p_b, L]$ of $D$ is done symmetrically.

\begin{figure}[t!]
    \centering
    \begin{subfigure}[b]{0.45\textwidth}
    \centering
        \includegraphics[scale=.75, page=8]{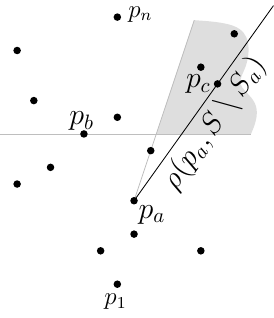}
         \subcaption{}
         \label{fig:fast-computation-good}
    \end{subfigure}
    \hfill
    \begin{subfigure}[b]{0.45\textwidth}
    \centering
        \includegraphics[scale=.75, page=9]{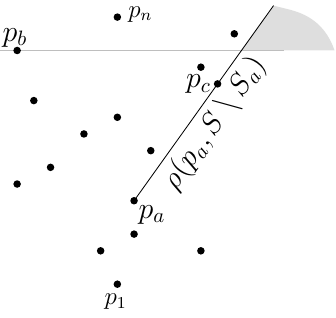}
         \subcaption{}
         \label{fig:fast-computation-bad}
    \end{subfigure}
    \caption{(a) For any  $b\in \{a+1,a+2,\dots,c-1\}$, we have that $p_c$ is in $\mathcal{R}(p_a,p_b)$. (b) For any  $b\in \{c,c+1,\dots,n\}$, we have that $\mathcal{R}(p_a,p_b)$ is empty. Region $\mathcal{R}(p_a,p_b)$ is shaded gray.}
    \label{fig:fast-computation}
\end{figure}

We compute the point $p_c$ with $c>a$ such that the ray $\rho(p_a,p_c)=\rho(p_a,S \setminus S_a)$ is the rightmost among the rays starting from $p_a$ and passing through a point in $S\setminus S_a$. This can be done in $\bigO(n)$ time by inspecting the points $p_{a+1},p_{a+2},\dots,p_{n}$. Then, we set $D[p_a,p_b, R] = \texttt{True}$ for all the points $p_b$ with $b=a+1,a+2,\dots,c-1$ and $D[p_a,p_b, R] = \texttt{False}$ for all the points $p_b$ with $b=c,c+1,\dots,n$. Indeed, for any  $b\in \{a+1,a+2,\dots,c-1\}$, we have that $p_c$ is in $\mathcal{R}(p_a,p_b)$, since it is strictly above the horizontal line through $p_b$ (given that $b<c$) and strictly to the right of the ray $\rho(p_a,p_b)$ (given that $\rho(p_a,p_c)$ is the rightmost among the rays starting from $p_a$ and passing through a point in $S\setminus S_a$); see \cref{fig:fast-computation-good}. Also, for any  $b\in \{c,c+1,\dots,n\}$, we have that $p_c$ is in $S_b \setminus S_a$, and, by definition of $p_c$, no point is strictly to the right of the ray $\rho(p_a,p_c)$, hence $\mathcal{R}(p_a,p_b)$ is empty; see \cref{fig:fast-computation-bad}.

\end{itemize}

This concludes the description of the $\bigO(n^2)$-time computation of the tables $C$ and $D$. Due to these tables, the computation performed by the enumeration algorithm at each node of the search tree $T$ takes $\bigO(1)$ time. Indeed, consider a node $\mu$ of $T$ associated to a safe bipath~$B$ composed of two monotone $st$-paths ending at the points $p_{\ell}$ and $p_r$. Let $i=\max\{\ell,r\}$. By means of the value $C[p_{\ell},p_{r},L]$ and $D[p_r,p_{i+1},R]$, we can respectively test in $\bigO(1)$ time whether $B':=B\cup \overline{p_{\ell} p_{i+1}}$ is a bipath and, in case it is, whether it is safe. If $B'$ is a safe bipath, then the algorithm adds to $\mu$ a child corresponding to $B'$, and the traversal continues on that child. Once the traversal backtracks to $\mu$ again, or if $B'$ was not a safe bipath in the first place, by means of the values $C[p_{\ell},p_{r},R]$ and $D[p_\ell,p_{i+1},L]$, we can respectively test in $\bigO(1)$ time whether $B'':=B\cup \overline{p_r p_{i+1}}$ is a bipath and, in case it is, whether it is safe. If $B''$ is a safe bipath, then the algorithm adds to $\mu$ a child corresponding to $B''$, and the traversal continues on that child. Since the computation at each node takes $\bigO(1)$ time and since $T$ has $n$ levels, it follows that the algorithm's delay is in $\bigO(n)$.

\cref{prop:level-leaves,prop:leaf-cycle,prop:distinct-cycle,prop:all-cycles,prop:space-time} complete the proof of \cref{th:hamiltonian-enumeration}.

\section{Conclusions and Open Problems}

We addressed basic pointset embeddability problems for upward planar graphs. We proved that UPSE testing is \NP-hard even for planar $st$-graphs composed of internally-disjoint $st$-paths and for directed trees composed of directed root-to-leaf paths. 
For planar $st$-graphs, we showed that \textsc{UPSE Testing} can be solved in $O(n^{4k})$ time, where $k$ is the maximum $st$-cutset of $G$, and we provided an algorithm to enumerate all UPSEs of $G$ on $S$ with $\bigO(n)$ worst-case delay. We also showed how to enumerate all monotone polygonalizations of a given pointset with $\bigO(n)$ worst-case delay. We point out the following open problems.

\begin{itemize}
    \item Our NP-hardness proofs for \textsc{UPSE testing} use the fact that the points are not in general position. Given a directed tree $T$ on $n$ vertices and a set $S$ of $n$ points {\em in general position}, is it \NP-hard to decide whether $T$ has an UPSE on $S$?
    \item Can \textsc{UPSE testing} be solved in polynomial time or does it remain \NP-hard if the input is a {\em maximal} planar $st$-graph?
    \item We proved that \textsc{UPSE testing} for a planar $st$-graph is in XP with respect to the size of the maximum $st$-cutset of $G$. Is the problem in FPT with respect to the same parameter? Are there other interesting parameterizations for the problem?
    \item Let $S$ be a pointset and $\mathcal P$ be a non-crossing path on a subset of $S$. Is it possible to decide in polynomial time whether $\mathcal P$ can be extended to a polygonalization of $S$? A positive answer would imply an algorithm with polynomial delay for enumerating the polygonalizations of a pointset, with the same approach as the one we adopted in this paper for monotone polygonalizations.    
\end{itemize}

\bibliographystyle{splncs04}
\bibliography{bibliography}

\end{document}